%% file: paper_theory_arXiv.tex
\documentclass[reprint,
aps,superscriptaddress,twocolumn,showpacs,longbibliography]{revtex4-2}

\usepackage{graphicx,color}
\usepackage{amsmath,amsfonts,enumerate,amsthm,amssymb,bbm}
\usepackage[colorlinks=true,citecolor=blue,linkcolor=magenta]{hyperref}
\usepackage[all]{hypcap}
\usepackage{multirow}
\usepackage{bbold}
\usepackage{braket}
\usepackage{soul}
\usepackage{float}
\usepackage[caption = false]{subfig}
\usepackage[normalem]{ulem}
\usepackage{scrextend}
\usepackage[usenames,dvipsnames]{xcolor}
\usepackage{dsfont}
\usepackage{nicefrac}
\usepackage[linesnumbered,ruled,vlined]{algorithm2e}
\usepackage{quantikz}
\usepackage{apptools}

\mathchardef\ordinarycolon\mathcode`\:
\mathcode`\:=\string"8000
\begingroup \catcode`\:=\active
  \gdef:{\mathrel{\mathop\ordinarycolon}}
\endgroup

\theoremstyle{plain}
\newtheorem{thm}{Theorem}
\newtheorem{corol}{Corollary}
\newtheorem{lem}{Lemma}

\theoremstyle{definition}
\newtheorem{defn}{Definition}
\theoremstyle{remark}
\newtheorem{rmk}{Remark}
\newtheorem{ex}{Example}

\AtAppendix{\counterwithin{thm}{section}}
\AtAppendix{\counterwithin{corol}{section}}
\AtAppendix{\counterwithin{lem}{section}}
\AtAppendix{\counterwithin{propos}{section}}
\AtAppendix{\counterwithin{defn}{section}}
\AtAppendix{\counterwithin{rmk}{section}}
\AtAppendix{\counterwithin{ex}{section}}

\def\<{\langle}

\def\I{ \mathbb{1} }

\def\I{ \mathbbm{1} }
\def\tr{ \mbox{tr} \,}
\def\>{\rangle}
\def\<{\langle}
\DeclareMathOperator{\Tr}{Tr} 
 
\DeclareMathOperator{\poly}{poly}

\newcommand{\alv}{\vec{\alpha}}

\newcommand{\norm}[1]{\left\lVert#1\right\rVert}

\renewcommand{\ket}[1]{|#1\rangle}               
\renewcommand{\bra}[1]{\langle #1|}              

\renewcommand{\vec}[1]{\boldsymbol{#1}}  

\newcommand{\stkout}[1]{\ifmmode\text{\sout{\ensuremath{#1}}}\else\sout{#1}\fi}
\newif\ifverbose
\verbosetrue

\begin{document}

\title{Out-of-distribution generalization for learning quantum dynamics}

\author{Matthias C.~Caro}
\thanks{The first two authors contributed equally to this work.}
\affiliation{Department of Mathematics, Technical University of Munich, Garching, Germany.}
\affiliation{Munich Center for Quantum Science and Technology (MCQST), Munich, Germany.}
\affiliation{Dahlem Center for Complex Quantum Systems, Freie Universität Berlin, Berlin, Germany.}
\affiliation{Institute for Quantum Information and Matter, Caltech, Pasadena, CA, USA.}

\author{Hsin-Yuan Huang}
\thanks{The first two authors contributed equally to this work.}
\affiliation{Institute for Quantum Information and Matter, Caltech, Pasadena, CA, USA.}
\affiliation{Department of Computing and Mathematical Sciences, Caltech, Pasadena, CA, USA.}

\author{Nicholas Ezzell}
\affiliation{Information Sciences, Los Alamos National Laboratory, Los Alamos, NM, USA.}
\affiliation{Department of Physics \& Astronomy, University of Southern California, Los Angeles, CA, USA.}

\author{Joe Gibbs}
\affiliation{Department of Physics, University of Surrey, Guildford, GU2 7XH, UK.}
\affiliation{AWE, Aldermaston, Reading, RG7 4PR, UK.}

\author{Andrew T. Sornborger} 
\affiliation{Information Sciences, Los Alamos National Laboratory, Los Alamos, NM, USA.}

\author{Lukasz Cincio}
\affiliation{Theoretical Division, Los Alamos National Laboratory, Los Alamos, NM, USA.}

\author{Patrick~J.~Coles} 
\affiliation{Theoretical Division, Los Alamos National Laboratory, Los Alamos, NM, USA.}
\affiliation{Normal Computing Corporation, New York, NY, USA.}

\author{Zo\"{e} Holmes}
\affiliation{Information Sciences, Los Alamos National Laboratory, Los Alamos, NM, USA.}
\affiliation{Institute of Physics, Ecole Polytechnique F\'{e}d\'{e}derale de Lausanne (EPFL), CH-1015 Lausanne, Switzerland.}

\date{\today}

\begin{abstract}
Generalization bounds are a critical tool to assess the training data requirements of Quantum Machine Learning (QML). Recent work has established guarantees for in-distribution generalization of quantum neural networks (QNNs), where training and testing data are drawn from the same data distribution. However, there are currently no results on out-of-distribution generalization in QML, where we require a trained model to perform well even on data drawn from a different distribution to the training distribution. Here, we prove out-of-distribution generalization for the task of learning an unknown unitary. In particular, we show that one can learn the action of a unitary on entangled states having trained only product states. Since product states can be prepared using only single-qubit gates, this advances the prospects of learning quantum dynamics on near term quantum hardware, and further opens up new methods for both the classical and quantum compilation of quantum circuits.
\end{abstract}

\maketitle

\section{Introduction}

In Quantum Machine Learning (QML) a quantum neural network (QNN) is trained using classical or quantum data, with the goal of learning how to make accurate predictions on unseen data~\cite{biamonte2017quantum,schuld2015introduction, schuld2021machine}. This ability to extrapolate from training data to unseen data is known as generalization. There is much excitement currently about the potential of such QML methods to outperform classical methods for a range of learning tasks~\cite{huang2021quantum,liu2021rigorous, aharonov2021quantum, huang2021information, chen2021exponential, chen2021hierarchy, cotler2021revisiting, huang2021power}. However, to achieve this, it is critical that the training data required for successful generalization can be produced efficiently.

While recent work has established a number of fundamental bounds on the \textit{amount} of training data required for successful generalization in QML~\cite{caro2020pseudo, abbas2020power, sharma2020reformulation, huang2021power, bu2021onthestatistical, banchi2021generalization, du2021efficient, gyurik2021structural, caro2021encodingdependent, caro2021generalization, chen2021expressibility, popescu2021learning, cai2022sample, Volkoff2021Universal}, less attention has been paid so far to the \textit{type} of training data required for generalization. In particular, prior work has established guarantees for the \textit{in-distribution generalization} of QML models, where training and testing data are assumed to be drawn independently from the same data distribution. However, in practice one may only have access to a limited type of training data, and yet be interested in making accurate predictions for a wider class of inputs. 
This is particularly an issue in the Noisy Intermediate-Scale Quantum (NISQ) era~\cite{preskill2018quantum}, when deep quantum circuits cannot be reliably executed, effectively limiting the quantum training data states that can be prepared.

In this article, we study \textit{out-of-distribution generalization} in QML. That is, we investigate generalization performance when the testing and training distributions do not coincide.
Specifically, we consider the task of learning unitary dynamics, which is a fundamental primitive for a range of QML algorithms. At its simplest, the target unitary could be the unknown dynamics of an experimental quantum system.  For this case, which has close links with quantum sensing~\cite{degen2017quantum} and Hamiltonian learning~\cite{wang2017experimental, Wiebe2014Quantum, Gentile2021Learning}, the aim is essentially to learn a digitalization of an analog quantum process. This could be performed using a `standard' quantum computer or a simpler experimental system with perhaps a limited gate set, as sketched in Fig.~\ref{fig:Schematic}a) and b) respectively.
Alternatively, the target unitary could take the form of a known gate sequence that one seeks to compile into a shorter depth circuit or a particular structured form~\cite{khatri2019quantum, sharma2019noise, Jones2022robustquantum,heya2018variational}. The compilation could be performed either on a quantum computer, see Fig.~\ref{fig:Schematic}c), or entirely classically, see Fig.~\ref{fig:Schematic}d). 
Such a subroutine can be used to reduce the resources required to implement larger scale quantum algorithms including those for dynamical simulations~\cite{cirstoiu2020variational,gibbs2021long, geller2021experimental, gibbs2022dynamical}.

Here we prove out-of-distribution generalization for unitary learning with a broad class of training and testing distributions. 
Specifically, we show that the average prediction error over any two \textit{locally scrambled}~\cite{kuo2020markovian, hu2021classical} ensembles of states are perfectly correlated up to a small constant factor. This is captured by our main theorem, Theorem~\ref{thm:comparing-different-locally-scrambled-costs}. By combining this observation with in-distribution generalization guarantees it follows that if the training and testing distributions are both locally scrambled (but potentially otherwise different distributions), out-of-distribution generalization is always possible between locally scrambled distributions. 
In particular, we show that a QNN trained on quantum data capturing the action of an efficiently implementable target unitary on a polynomial number of random product states, generalizes to test data composed of fully random states. 
That is, rather intriguingly, we show that one can learn the action of such a unitary on a broad spread of highly entangled states having only studied its action on a limited number of product states.

\begin{figure}[t]
\centering
\includegraphics[width =\columnwidth]{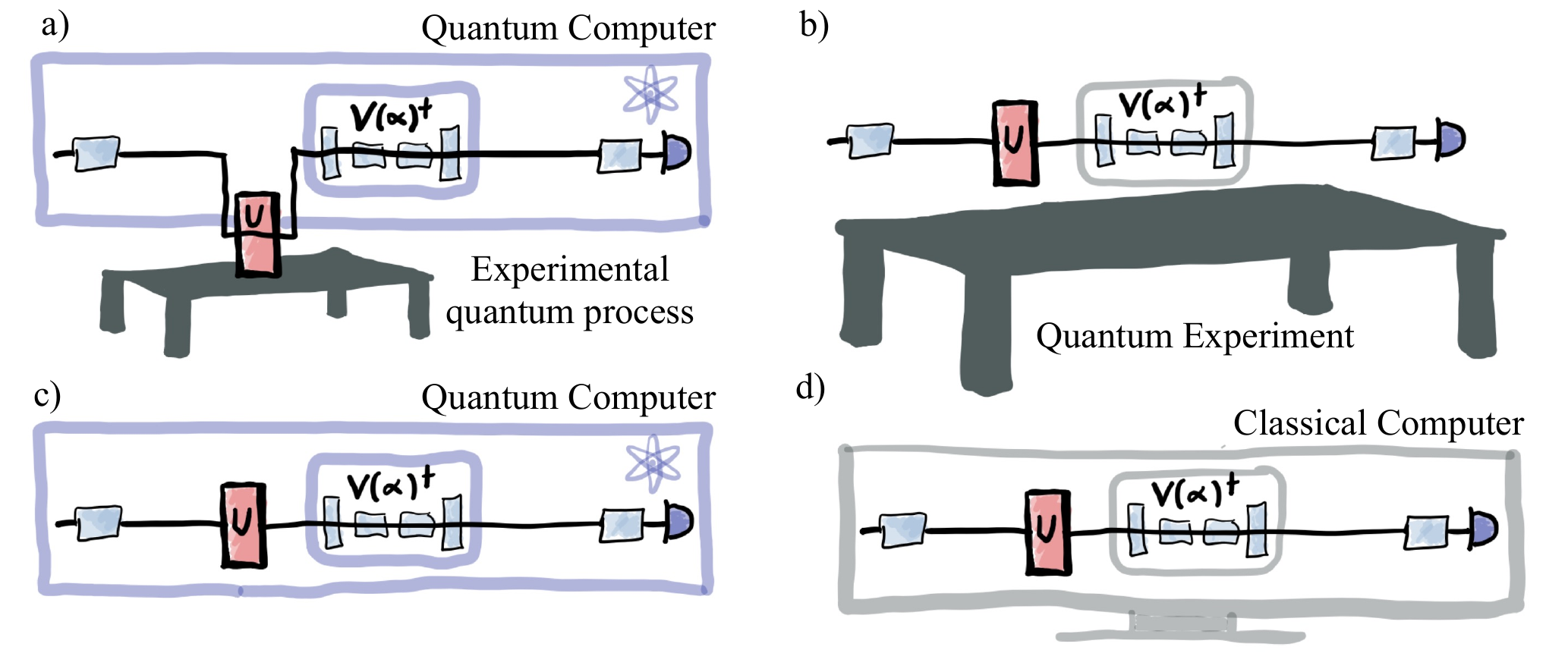}
\vspace{-6mm}
\caption{\textbf{Applications of Quantum Dynamics Learning.} a) Quantum dynamics learning of an experimental process using a quantum computer. b) Quantum dynamics learning with a more specialized experimental system with potentially a limited gate set. c) and d) Quantum compilation of a known unitary on a quantum computer and classical computer, respectively.}
\label{fig:Schematic}
\end{figure}

We numerically illustrate these analytical results by showing that the short time evolution of a Heisenberg spin chain can be well learned using only product state training data. Namely, we find that the out-of-distribution generalization error nearly perfectly correlates with the in-distribution generalization error and the training cost. 
In particular, in our numerical experiments, the testing performances achieved by the QML model on Haar-random states and on random product states differ only by a small constant factor, as predicted analytically.
We further perform noisy simulations that demonstrate how the noise accumulated preparing highly entangled states can prohibit training. In contrast, noisy training on product states, which can be prepared using only single-qubit gates, remains feasible. Additionally, in Appendix \ref{ap:additional-numerics} we numerically validate our generalization guarantees in a task of learning so-called fast scrambler unitaries \cite{belyansky2020minimal}. Thus our results make the possibility of using QML to learn unitary processes nearer term.

Our results further suggest a new quantum-inspired classical approach to unitary compilation. Namely, our results imply that a low-entangling unitary can be compiled using only low-entangled training states.
Such circuits can be readily simulated using classical tensor network methods, and hence this compilation can be performed classically.

\section{Results}

\subsection{Framework}

In this work we consider the QML task of learning an unknown $n$-qubit unitary $U\in\mathcal{U}((\mathbb{C}^{2})^{\otimes n})$. 
The goal is to use training states to optimize the classical parameters $\alv$ of $V(\alv)$, an $n$-qubit unitary QNN (or classical representation of a QNN), such that, for the optimized parameters $\alv_{\rm opt}$, $V(\alv_{\rm opt})$ well
predicts the action of $U$ on previously unseen test states. 

To formalize this notion of learning, we employ the framework of statistical learning theory~\cite{vapnik1971uniform, valiant1984theory}.
The prediction performance of the trained QNN $V(\alv_{\rm opt})$ can be quantified in terms of the average distance between the output state predicted by $V(\alv_{\rm opt})$ and the true output state determined by $U$. 
The average is taken over input states from a testing ensemble, which represents the ensemble of states that one wants to be able to predict the action of the target unitary on. 
More precisely, the goal is to minimize an \emph{expected risk}
\small
\begin{equation}\label{eq:test-cost-general}
    R_{\mathcal{P}}(\vec{\alpha}) 
    = \frac{1}{4}\mathbb{E}_{\ket{\Psi}\sim \mathcal{P}} \left[\big\lVert U\ket{\Psi}\bra{\Psi}U^\dagger  - V(\alv) \ket{\Psi}\bra{\Psi} V(\alv)^\dagger \big\rVert_{1}^2\right] \, ,
\end{equation}
\normalsize
where the testing distribution $\mathcal{P}$ is a probability distribution over (pure) $n$-qubit states $\ket{\Psi}$ and the factor of $1/4$ ensures $0 \leq R_{\mathcal{P}}(\vec{\alpha}) \leq 1$.

A learner will not have access to the full testing ensemble $\mathcal{P}$ and so cannot evaluate the cost in Eq.~\eqref{eq:test-cost-general}.
Instead, it is typically assumed that the learner has access to a training data set consisting of input-output pairs of pure $n$-qubit states, 
\begin{equation}\label{eq:training-data-general}
    \mathcal{D}_{\mathcal{Q}}(N) = \{(\ket{\Psi^{(j)}}, U\ket{\Psi^{(j)}}) \}_{j=1}^{N} \, ,
\end{equation}
where the $N$ input states are drawn independently from a training distribution $\mathcal{Q}$. Equipped with such training data, the learner may evaluate the \emph{training cost}
\small
\begin{equation}\label{eq:training-cost-general}
    \begin{split}
        C_{\mathcal{D}_{\mathcal{Q}}(N)}(\vec{\alpha})
        = \frac{1}{4 N} \sum_{j=1}^N \big\lVert &U\ket{\Psi^{(j)}}\bra{\Psi^{(j)}}U^\dagger -\\ &V(\vec{\alpha})  \ket{\Psi^{(j)}}\bra{\Psi^{(j)}} V(\vec{\alpha}) ^\dagger \big\rVert_{1}^2 \, .
    \end{split}
\end{equation}
\normalsize
We note that this cost can be rewritten in terms of the average fidelity as
\begin{equation}\label{eq:costfid}
        C_{\mathcal{D}_{\mathcal{Q}}(N)}(\vec{\alpha}) 
        = 1 -\frac{1}{N}\sum\limits_{j=1}^N  | \bra{\Psi^{(j)}}   V(\alv)^\dagger U \ket{\Psi^{(j)}} |^2  \, 
\end{equation}
and thus can be efficiently computed using a Loschmidt echo~\cite{sharma2020reformulation} or swap test circuit~\cite{buhrman2001quantum, gottesman2001quantum}. The hope is that by training the parameters $\alv$ of the QNN to minimize the training cost $C_{\mathcal{D}_{\mathcal{Q}}(N)}( \vec{\alpha})$ one will also achieve small risk $R_{\mathcal{P}}(\vec{\alpha})$.

However, whether such a strategy is successful crucially depends on whether the training cost $C_{\mathcal{D}_{\mathcal{Q}}(N)}( \vec{\alpha})$ is indeed a good proxy for the expected cost $R_{\mathcal{P}}\left( \alv \right)$.
This is exactly the question of \emph{generalization}: Does good performance on the training data imply good performance on (previously unseen) testing data?

In statistical learning theory, answers to this question are given in terms of \emph{generalization bounds}. 
These are bounds on the generalization error, which is typically taken to be the difference between expected risk and training cost, i.e., 
\begin{equation}
    \operatorname{gen}_{\mathcal{P}, \mathcal{D}_{\mathcal{Q}}(N)} \left(\alv_{\rm opt} \right) := R_\mathcal{P} \left(\alv_{\rm opt} \right) - C_{\mathcal{D}_{\mathcal{Q}}(N)} \left(\alv_{\rm opt} \right) \, .
\end{equation}
Usually, such bounds are proved under an i.i.d.~assumption on training and testing. That is, they are based on the assumptions (a) that the training examples are drawn independently from a training distribution $\mathcal{Q}$ and (b) that the training and testing distributions coincide, $\mathcal{Q}= \mathcal{P}$. In this case, we speak of \emph{in-distribution generalization}. 

In this paper, we consider \emph{out-of-distribution generalization} where we drop assumption (b) by allowing $\mathcal{Q} \neq \mathcal{P}$. 
Borrowing classical machine learning terminology, one can also regard this as a scenario of dataset shift \cite{quinonero2008dataset}, or more specifically covariate shift \cite{shimodaira2000improving, shen2021towards}, which is often addressed using transfer learning techniques~\cite{pratt1991direct, pan2009survey}.
We formulate our results for a broad class of ensembles called \emph{locally scrambled} ensembles.
In loose terms, locally scrambled ensembles of states can be thought of as ensembles of states that are \textit{at least locally random}~\footnote{Note, we use the terms `distribution' and `ensemble' interchangeably.}. More formally, they are defined as follows. 

\begin{defn}[Locally scrambled ensembles]\label{definition:locally-scrambled-ensembles}
    An ensemble of $n$-qubit unitaries is called \emph{locally scrambled} if it is invariant under pre-processing by tensor products of arbitrary local unitaries. 
    That is, a unitary ensemble $\mathcal{U}_{\rm LS}$ is locally scrambled iff for $U \sim \mathcal{U}_{\rm LS}$ and for any fixed $U_1,\ldots, U_n\in\mathcal{U}(\mathbb{C}^2)$ also $U (\bigotimes_{i=1}^n U_i)\sim\mathcal{U}_{\rm LS}$~\footnote{Here and elsewhere, the ``$\sim$'' notation means that the random variable on the left has the distribution on the right as its law. For instance, $U \sim \mathcal{U}_{\rm LS}$ means that the random unitary $U$ is drawn from the distribution $\mathcal{U}_{\rm LS}$.}.    Accordingly, an ensemble $\mathcal{S}_{\rm LS}$ of $n$-qubit quantum states is locally scrambled if it is of the form $\mathcal{S}_{\rm LS}=\mathcal{U}_{\rm LS}\ket{0}^{\otimes n}$ for some locally scrambled unitary ensemble $\mathcal{U}_{\rm LS}$~\footnote{Here and elsewhere, $\mathcal{U}\ket{0}^{\otimes n}$ denotes the ensemble of states generated by drawing unitaries from $\mathcal{U}$ and applying them to the $n$-qubit all-zero state $\ket{0}^{\otimes n}$.}. 
    We denote the classes of locally scrambled ensembles of unitaries and states as $\mathbb{U}_{\rm LS}$ and $\mathbb{S}_{\rm LS}$, respectively. 
\end{defn}

In fact, our results hold for a slightly broader class of ensembles where we only require that the ensemble agrees with a locally scrambled one up to and including its (complex) second moments. That is, more informally, the average over the ensemble agrees with those of a locally scrambled ensemble over all functions of $U$ that contain at most two copies of $U$. We will denote these broader classes of unitary and state ensembles, which we formally define in Appendix~\ref{app:Prelim}, as $\mathbb{U}_{\rm LS}^{(2)}$ and $\mathbb{S}_{\rm LS}^{(2)}$, respectively.

\medskip

In our results, we suppose that both the testing and training ensembles are such ensembles, i.e. $\mathcal{P} \in \mathbb{S}_{\rm LS}^{(2)}$ and $\mathcal{Q} \in \mathbb{S}_{\rm LS}^{(2)}$. 
However, as $\mathbb{S}_{\rm LS}^{(2)}$ captures a variety of different possible ensembles, $\mathcal{P}$ and $\mathcal{Q}$ can be ensembles containing very different sorts of states. 
In particular, as detailed further in Appendix~\ref{app:Prelim}, the following are important examples of ensembles in $\mathbb{S}_{\rm LS}^{(2)}$:

\begin{itemize}
    \item  $\mathcal{S}_{{\rm Haar}_1^{\otimes n}}$ - Products of Haar-random single-qubit states.
    \item $\mathcal{S}_{{\rm Stab}_1^{\otimes n}}$ - Products of random single-qubit stabilizer states.
    \item  $\mathcal{S}_{{\rm Haar}_k^{\otimes n/k}}$ - Products of Haar-random $k$-qubit states.
    \item $\mathcal{S}_{{\rm Haar}_n}$ - Haar-random $n$-qubit states.
    \item $\mathcal{S}_{\rm 2design}$ - A $2$-design on $n$-qubit states.
    \item $\mathcal{S}_{\rm RandCirc}^{\mathcal{A}_k}$ - The output states of random quantum circuits. (Here $\mathcal{A}_k$ denotes the $k$-local $n$-qubit quantum circuit architecture from which the random circuit is constructed.) 
\end{itemize}

These examples highlight that the class of locally scrambled ensembles includes both ensembles that consist solely of product states and ensembles composed mostly of highly entangled states. We can use this to our advantage to construct more efficient machine learning strategies.

Typically the learner will be interested in learning the action of a unitary on a wide class of input states including both entangled and unentangled states. For example, they might be interested in learning the action of a unitary on all states that can be efficiently prepared on a quantum computer using a polynomial-depth hardware-efficient layered ansatz. 
Thus in general the expected risk should be evaluated over distributions such as $\mathcal{S}_{{\rm Haar}_n}$, $\mathcal{S}_{\rm 2design}$ or $\mathcal{S}_{\rm RandCirc}^{\mathcal{A}_k}$ (for $k \geq 2$) which cover a large proportion of the total Hilbert space.

In classical machine learning one often thinks of the training data as given. However, in the context of learning or compiling quantum unitary dynamics (as sketched in Fig.~\ref{fig:Schematic}), one in practice needs either to prepare the training states on a quantum computer or in an experimental setup, or to be able to efficiently simulate them classically. Thus, it is desirable to train on states that can be prepared using simple circuits, i.e., those that are short depth, low-entangling or require only simple gates. This is especially important in the NISQ era due to noise-induced barren plateaus~\cite{wang2020noise} or other noise-related issues~\cite{franca2020limitations}.
Therefore, as random stabilizer states and random product states can be prepared using only a single layer of single-qubit gates, it makes practical sense to train using the ensembles $\mathcal{S}_{{\rm Haar}_1^{\otimes n}}$ or $\mathcal{S}_{{\rm Stab}_1^{\otimes n}}$. 

In this manner the class of ensembles that are locally scrambled to the second moment, $\mathbb{S}_{\rm LS}^{(2)}$, divides naturally into sub-classes of ensembles that give rise to training sets and testing sets. We sketch this in Fig.~\ref{fig:LocallyScramblingVenn}.

\begin{figure}[t]
\centering
\includegraphics[width =\columnwidth]{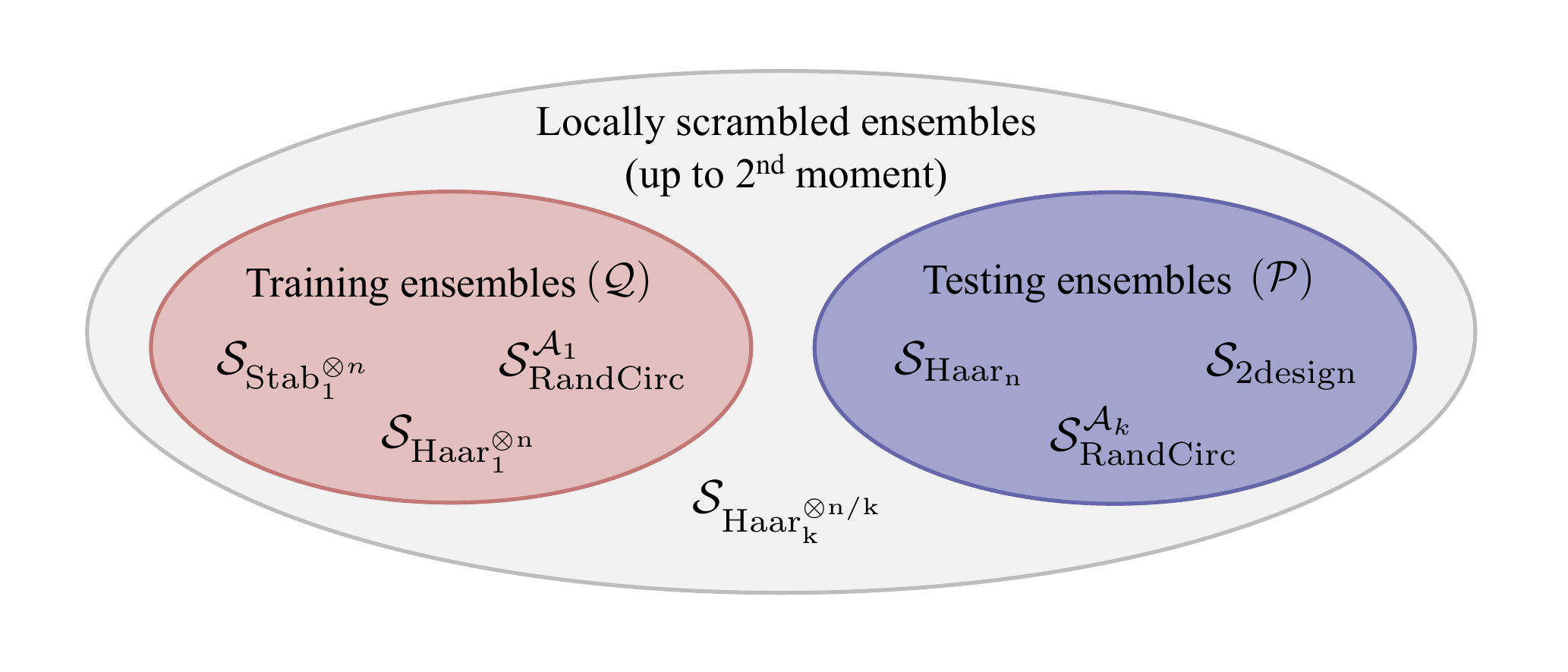}
\vspace{-6mm}
\caption{\textbf{Locally Scrambled Ensembles.} Venn diagram showing how the class of ensembles that are locally scrambled up to the second moment, $\mathbb{S}_{\rm LS}^{(2)}$, divides naturally into training ensembles and testing ensembles. For the formal definitions of each of the ensembles referenced see Appendix~\ref{app:Prelim}.}
\label{fig:LocallyScramblingVenn}
\end{figure}

\subsection{Analytical Results}\label{sec:analytical-results}

Having set up our framework, we now present our analytical results. 
First, we show that all locally scrambled ensembles lead to closely related expected risks for unitary learning.
Second, we use this observation to lift in-distribution generalization to out-of-distribution generalization when using a QNN to learn an unknown unitary from quantum data. For the formal proofs see Appendix~\ref{app:AnalyticResults}.

\subsubsection{Equivalence of Locally Scrambled Risks}

We first show a close connection between the risks for unitary learning arising from any locally scrambled ensembles. More precisely, we show that they can be upper and lower bounded in terms of the expected risk over the Haar distribution in our main technical result:

\begin{lem}\label{lem:locally-scrambled-cost-versus-hst-cost}
    For any $\mathcal{Q} \in \mathbb{S}_{\rm LS}^{(2)}$ and any parameter setting $\alv$,
    \begin{equation}
        \frac{1}{2} R_{\mathcal{S}_{{\rm Haar}_n}}(\vec{\alpha})
        \leq \frac{d}{d+1} R_{\mathcal{Q}} \left( \vec{\alpha} \right)
        \leq R_{\mathcal{S}_{{\rm Haar}_n}}(\vec{\alpha}) \, ,
    \end{equation}
    where $d = 2^n$ is the dimension of the target unitary $U$ being learned. 
\end{lem}

This result establishes that learning over any locally scrambled distribution is effectively equivalent (up to a constant multiplicative factor) to learning over the uniform distribution over the entire Hilbert space.
We note that the factor of $1/2$ in the lower bound emerges from the structure of our proof, and for typical cases we expect the relation between the costs to be tighter still. 
We explore this numerically in Appendix~\ref{ap:TightnessNumerics} for the special case of training on random product states, i.e. $\mathcal{Q} = \mathcal{S}_{\rm Haar_1^{\otimes n}}$.

A direct consequence of Lemma~\ref{lem:locally-scrambled-cost-versus-hst-cost} is that the risks arising from any two locally scrambled ensembles are related as follows. 

\begin{thm}[Equivalence of locally scrambled ensembles for comparing unitaries]\label{thm:comparing-different-locally-scrambled-costs}
    Let $\mathcal{P} \in \mathbb{S}_{\rm LS}^{(2)}$ and $\mathcal{Q} \in \mathbb{S}_{\rm LS}^{(2)}$, then for any parameter setting $\alv$,
    \begin{equation}
        \frac{1}{2} R_{\mathcal{Q}} \left( \vec{\alpha} \right)
        \leq R_{\mathcal{P}} \left( \vec{\alpha}\right)
        \leq 2 R_{\mathcal{Q}} \left(\vec{\alpha} \right) \, .
    \end{equation}
\end{thm}

Theorem~\ref{thm:comparing-different-locally-scrambled-costs} establishes an equivalence (up to a constant multiplicative factor) between all locally scrambled testing distributions for the task of learning an unknown unitary on average. In particular, even simple locally scrambled ensembles, such as tensor products of Haar-random single-qubit states or of random single-qubit stabilizer states, are for this purpose effectively equivalent to seemingly more complex locally scrambled ensembles. The latter include the output states of random quantum circuits or, indeed, globally Haar-random states.

\subsubsection{Out-of-Distribution Generalization for QNNs Trained on Locally Scrambled States}

Theorem~\ref{thm:comparing-different-locally-scrambled-costs} gives rise to a general template for lifting in-distribution generalization bounds for QNNs to out-of-distribution generalization guarantees in unitary learning. This is captured by the following corollary:

\begin{corol}[Locally scrambled out-of-distribution generalization from in-distribution generalization]\label{corol:ood-generalization-from-id-generalization-unitary-learning}
    Let $\mathcal{P} \in \mathbb{S}_{\rm LS}^{(2)}$ and $\mathcal{Q} \in \mathbb{S}_{\rm LS}^{(2)}$.
    Let $U$ be an unknown $n$-qubit unitary.
    Let $V(\alv)$ be an $n$-qubit unitary QNN that is trained using training data ${\mathcal{D}_{\mathcal{Q}}(N)}$ containing $N$ input-output pairs, with inputs drawn from the ensemble $\mathcal{Q}$. 
    Then, for any parameter setting $\alv$,
    \begin{equation}
        R_{\mathcal{P}} (\alv)
        \leq 2\left( C_{\mathcal{D}_{\mathcal{Q}}(N)} (\alv) + \operatorname{gen}_{\mathcal{Q}, \mathcal{D}_{\mathcal{Q}}(N)} \left(\alv \right) \right) \, .
    \end{equation}
\end{corol}

Thus, when training using training data $\mathcal{D}_{\mathcal{Q}}(N)$, the out-of-distribution risk $R_{\mathcal{P}} (\alv_{\rm opt})$ of the optimized parameters $\alv_{\rm opt}$ after training is controlled in terms of the optimized training cost $C_{\mathcal{D}_{\mathcal{Q}}(N)} (\alv_{\rm opt})$ and the in-distribution generalization error $\operatorname{gen}_{\mathcal{Q}, \mathcal{D}_{\mathcal{Q}}(N)} \left(\alv_{\rm opt} \right)$.
We can now bound the in-distribution generalization error using already known QML in-distribution generalization bounds~\cite{caro2020pseudo, abbas2020power, huang2021power, bu2021onthestatistical, sharma2020reformulation, banchi2021generalization, du2021efficient, gyurik2021structural, caro2021encodingdependent, caro2021generalization, chen2021expressibility, popescu2021learning, cai2022sample} (or, indeed, any such bounds that are derived in the future). 
We point out that our results up to this point do not require any assumptions on the QNN architecture underlying $V(\alv)$, except for overall unitarity.
As a concrete example of guarantees that can be obtained this way, we combine Corollary~\ref{corol:ood-generalization-from-id-generalization-unitary-learning} with an in-distribution generalization bound established in~\cite{caro2021generalization} to prove:

\begin{corol}[Locally scrambled out-of-distribution generalization for QNNs]\label{corol:locally-scrambled-ood-generalization-qnn}
    Let $\mathcal{P} \in \mathbb{S}_{\rm LS}^{(2)}$ and $\mathcal{Q} \in \mathbb{S}_{\rm LS}^{(2)}$.
    Let $U$ be an unknown $n$-qubit unitary.
    Let $V(\alv)$ be an $n$-qubit unitary QNN with $T$ parameterized local gates.
    When trained with the cost $C_{\mathcal{D}_{\mathcal{Q}}(N)}$ using training data $\mathcal{D}_{\mathcal{Q}}(N)$, the out-of-distribution risk w.r.t.~$\mathcal{P}$ of the parameter setting $\alv_{\rm opt}$ after training satisfies
    \small
    \begin{equation}
        R_{\mathcal{P}} (\alv_{\rm opt})
        \leq 2 C_{\mathcal{D}_{\mathcal{Q}}(N)} (\alv_{\rm opt}) + \mathcal{O} \left( \sqrt{\frac{T \log (T)}{N}}\right)\, ,
    \end{equation}
    \normalsize
    with high probability over the choice of training data of size $N$ according to $\mathcal{Q}$.
\end{corol}

The out-of-distribution generalization guarantee of Corollary~\ref{corol:locally-scrambled-ood-generalization-qnn} is particularly interesting if the training data is drawn from a distribution composed only of products of single-qubit Haar-random or random stabilizer states, i.e. $\mathcal{Q} = \mathcal{S}_{{\rm Haar}_1^{\otimes n}}$ or $\mathcal{Q} = \mathcal{S}_{{\rm Stab}_1^{\otimes n}}$, but the testing data is drawn from more complex distributions such as the Haar ensemble or the outputs of random circuits, i.e. $\mathcal{P} = \mathcal{S}_{\rm Haar_n}$ or $\mathcal{P} = \mathcal{S}_{\rm RandCirc}$. 
In this case, Corollary~\ref{corol:locally-scrambled-ood-generalization-qnn} implies that efficiently implementable unitaries can be learned using a small number of simple unentangled training states. More precisely, if $U$ can be approximated via a QNN with $\poly (n)$ trainable local gates, then only $\poly (n)$ unique product training states suffice to learn the action of $U$ on the Haar distribution, i.e. across the entire Hilbert space. 

To understand why out-of-distribution generalization is possible, recall that any state is linearly spanned by $n$-qubit Pauli observables $P \in \{I, X, Y, Z\}^{\otimes n}$, and each Pauli observable $P$ can be written as a linear combination of product states $\ket{s}\!\bra{s} = \bigotimes_{i=1}^n \ket{s_i}\!\bra{s_i}$, where $s_i \in \{0, 1, +, -, y+, y-\}$.
These two facts imply that for any state $\ket{\phi}\!\bra{\phi}$, there exists coefficients $\alpha_s$, such that $\ket{\phi}\!\bra{\phi} = \sum_{s} \alpha_s \ket{s}\!\bra{s}$.
Hence, if we exactly know $U \ket{s}\!\bra{s} U^\dagger$ for all $6^n$ product states $\ket{s}\!\bra{s}$, then we can figure out $U \ket{\phi}\!\bra{\phi} U^\dagger$ for any state $\ket{\phi}\!\bra{\phi}$ by linear combination.
However, this requires an exponential number of product states in the training data.
In our prior work~\cite{caro2021generalization}, we show that one only needs $\poly (n)$ training product states to approximately know $U \ket{s}\!\bra{s} U^\dagger$ for most of the $6^n$ product states, assuming $U$ is efficiently implementable.
The key insight in this work is that one can predict $U \ket{\phi}\!\bra{\phi} U^\dagger$ as long as the coefficients $\alpha_s$ in $\ket{\phi}\!\bra{\phi} = \sum_{s} \alpha_s \ket{s}\!\bra{s}$ are sufficiently random and spread out across the $6^n$ product states.
We make this condition precise by defining locally scrambled ensembles and proving that the action of $U$ on a state sampled from any such ensemble can be predicted.
In Appendix \ref{ap:role-of-linearity}, we further discuss the role that linearity plays in our results.

\medskip

We can immediately extend our results for out-of-distribution to local variants of costs. Such local costs are essential to avoid cost-function-dependent barren plateaus~\cite{cerezo2021cost} when training a shallow QNN.
As a concrete example, when taking $\mathcal{S}_{{\rm Haar}_1^{\otimes n}}$ as the training ensemble, we can consider the local training cost
\small
\begin{equation}\label{eq:product-local-training-cost}
\begin{aligned}
        &C_{\rm \footnotesize Prod, N}^{L}(\alv) 
        \\ &= 1 -\frac{1}{N}\sum\limits_{j=1}^N \Tr \left[  \ket{\Psi_{\rm \footnotesize Prod}^{(j)}}\bra{\Psi_{\rm \footnotesize Prod}^{(j)}}  U^\dagger V(\alv) H_{\rm L}^{(j)} V(\alv)^\dagger U\right] \, ,
\end{aligned}
\end{equation}
\normalsize
where $\ket{\Psi_{\rm \footnotesize Prod}^{(j)}} = \bigotimes_{i=1}^n \ket{\psi_i^{(j)} } \sim \mathcal{S}_{{\rm Haar}_1^{\otimes n}}$ for all $j$ and we have introduced the local measurement operator $ H_{\rm L}^{(j)} = \frac{1}{n} \sum_{i=1}^n \ket{\psi_i^{(j)}} \bra{\psi_i^{(j)} }\otimes \mathds{1}_{\bar{i}}$. This local cost is faithful to its global variant for product state training in the sense that it vanishes under the same conditions~\cite{khatri2019quantum}, but crucially, in contrast to the global case, may be trainable~\cite{cerezo2021cost}.
In Appendix~\ref{ap:genproofs}, we prove a version of Corollary~\ref{corol:locally-scrambled-ood-generalization-qnn} when training on the local cost from Eq.~\eqref{eq:product-local-training-cost}. Specifically we find: 
\begin{corol}[Locally scrambled out-of-distribution generalization for QNNs via a local cost]\label{corol:locally-scrambled-ood-generalization-qnn-local-cost}
    Let $\mathcal{P} \in \mathbb{S}_{\rm LS}^{(2)}$ and
    let $U$ be an unknown $n$-qubit unitary.
    Let $V(\alv)$ be an $n$-qubit unitary QNN with $T$ parameterized local gates.
    When trained with the cost $C^{L}_{\rm \footnotesize Prod, N}$, the out-of-distribution risk w.r.t.~$\mathcal{P}$ of the parameter setting $\alv_{\rm opt}$ after training satisfies
    \small
    \begin{equation}
    \begin{aligned}
         R_{\mathcal{P}} (\alv_{\rm opt})
        \leq \, &2 n C^{L}_{\rm \footnotesize Prod, N}(\alv_{\rm opt}) + \mathcal{O} \left( n\sqrt{\frac{T \log (T)}{N}} \right)\, , 
    \end{aligned}
    \end{equation}
    \normalsize
    with high probability over the choice of training data of size $N$.
\end{corol}
Clearly, analogous local variants of the training cost can be defined whenever the respective ensemble has a tensor product structure (such as $\mathcal{S}_{{\rm Stab}_1^{\otimes n}}$).  However, if the training data is highly entangled, constructing such local costs in this manner is not possible. Thus, this is another important consequence of our results: The ability to train solely on product state inputs makes it straightforward to generate the local costs that are necessary for efficient training. 

The results presented thus far concern the number of unique training states required for generalization, but in practice multiple copies of each training state will be needed for successful training. As $\mathcal{O}(1/\epsilon^2)$ shots are required to evaluate a cost to precision $\epsilon$ and since for gradient based training methods one needs to evaluate the partial derivative of the cost with respect to each of the $T$ trainable parameters, one would expect to need on the order of $\mathcal{O}(T M_{\rm opt}/\epsilon^2)$ copies of each of the $N$ input states and output states to reduce the cost to $\mathcal{\epsilon}$.  
Here $M_{\rm opt}$ is the number of optimization steps.
Classical shadow tomography~\cite{huang2020predicting, elben2022randomized, huang2021provably} provides a way towards a copy complexity bound that is independent of the number of optimization steps.
Namely, exploiting covering number bounds for the space of pure output states of polynomial-size quantum circuits (compare~\cite{caro2021generalization, huang2021quantum}), polynomial-size classical shadows can be used to perform tomography among such states. 
In the case of an efficiently implementable target unitary $U$ and QNN $V(\alv)$ that both admit a circuit representation with $T\in\mathcal{O}(\poly (n))$ local gates, $\mathcal{O} (T\log(T/\epsilon)/\epsilon^2)\leq\Tilde{\mathcal{O}}(\poly (n)/\epsilon^2)$ copies of each of the input states $\ket{\Psi^{(j)}}$ and output states $\ket{\Phi^{(j)}}$ suffice to approximately evaluate the cost (both the global and local variants) and its partial derivatives arbitrarily often.

\medskip 

\subsection{Numerical Results}\label{sec:numerical-results}

Here we provide numerical evidence to support our analytical results showing that out-of-distribution generalization is possible for the learning of quantum dynamics. We focus on the task of learning the parameters of an unknown target Hamiltonian by studying the evolution of product states under it. 

For concreteness, we suppose that the target Hamiltonian is of the form
\small
\begin{equation}\label{eq:Hamiltonian}
    H(\vec{p},\vec{q},\vec{r}) = \sum_{k=1}^{n-1} \left( Z_k Z_{k+1} + p_k X_k X_{k+1} \right) + \sum_{k=1}^{n}\left( q_k X_k + r_k Z_k \right) \ ,
\end{equation}
\normalsize
with the specific parameter setting $(\vec{p}^\ast, \vec{q}^\ast, \vec{r}^\ast)$ given by $p_k^\ast = \sin \left( \frac{\pi k}{2n} \right)$ for $1 \leq k \leq n-1$ and $q_k^\ast = \sin  \left( \frac{\pi k}{n} \right)$, $r_k^\ast = \cos \left( \frac{\pi k}{n} \right)$ for $1 \leq k \leq n$.
The learning is performed by comparing the exact evolution under $e^{-i H(\vec{p}^\ast,\vec{q}^\ast,\vec{r}^\ast) t}$ to a Trotterized ansatz. Specifically, we use an $L$ layered ansatz
$V_L(\vec{p},\vec{q},\vec{r}) := \left( U_{\Delta t}(\vec{p},\vec{q},\vec{r})\right)^L$ where $U_{\Delta t}$ is a second order Trotterization of $e^{-i H(\vec{p},\vec{q},\vec{r}) \Delta t}$. That is,
\small
\begin{equation} \label{eq:soT}
    U_{\Delta t}(\vec{p},\vec{q},\vec{r}) 
    = e^{-iH_A(\vec{r}) \Delta t/2} e^{-iH_B(\vec{p},\vec{q}) \Delta t} e^{-iH_A(\vec{r}) \Delta t/2}   
\end{equation}
\normalsize
where the Hamiltonians $H_A(\vec{r}) := \sum_{k=1}^{n-1} Z_k Z_{k+1} + \sum_{k=1}^n r_k Z_k $ and $H_B(\vec{p}, \vec{q}) := \sum_{k=1}^{n-1} p_k X_k X_{k+1} + \sum_{k=1}^n q_k X_k$ contain only commuting terms and so can be readily exponentiated.

We attempt to learn the vectors $\vec{p}^\ast$, $\vec{q}^\ast$, and $\vec{r}^\ast$ by comparing $e^{-i H(\vec{p}^\ast,\vec{q}^\ast,\vec{r}^\ast) t}\ket{\psi_j}$ and $V_L(\vec{p},\vec{q},\vec{r})\ket{\psi_j}$ over $N$ random product states $\ket{\psi_j}$. To do so, we use the training data $ \mathcal{D}_{\mathcal{Q}}(N)$ with $\mathcal{Q} = \mathcal{S}_{{\rm Haar}_1^{\otimes n}}$, and the cost function given in Eq.~\eqref{eq:costfid}.
The learning is performed classically for $n = 4, \ldots, 12$ and $L = 2, \ldots, 5$ and we take the total evolution time to be $t=0.1$. For all values of $n$ we train on two product states, i.e. $N =2$. We repeated the optimization 5 times in each case and kept the best run. While the small training data size $N=2$ was sufficient for the model considered here, in Appendix \ref{ap:additional-numerics} we present a more involved unitary learning setting that requires larger values of $N$.

Fig.~\ref{fig:main_numerics} plots the in-distribution risk and out-of-distribution risk as a function of the final optimized cost function values, $C_{\mathcal{D}_{\mathcal{Q}}(2)}(\vec{\alpha_{\rm opt}})$ with $\mathcal{Q} = \mathcal{S}_{{\rm Haar}_1^{\otimes n}}$. Here the in-distribution risk is the average prediction error over random product states, i.e. $R_{\mathcal{S}_{{\rm Haar}_1^{\otimes n}}}$, and for the out-of-distribution testing we chose to compute the risk over the global Haar distribution, i.e. $R_{\mathcal{S}_{{\rm Haar}_n}}$. 
These risks can be evaluated analytically using Lemma \ref{lem:locally-scrambled-cost-explicit-expression} and Eqs.~\eqref{eq:hst-cost}, \eqref{eq:hst-cost-versus-expected-haar-cost}.
The linear correlation between the cost function and both $R_{\mathcal{S}_{{\rm Haar}_1^{\otimes n}}}$ and $R_{\mathcal{S}_{{\rm Haar}_n}}$ demonstrates that both in-distribution and out-of-distribution generalization have been successfully achieved.

\begin{figure}[t]
\centering
\includegraphics[width =\columnwidth]{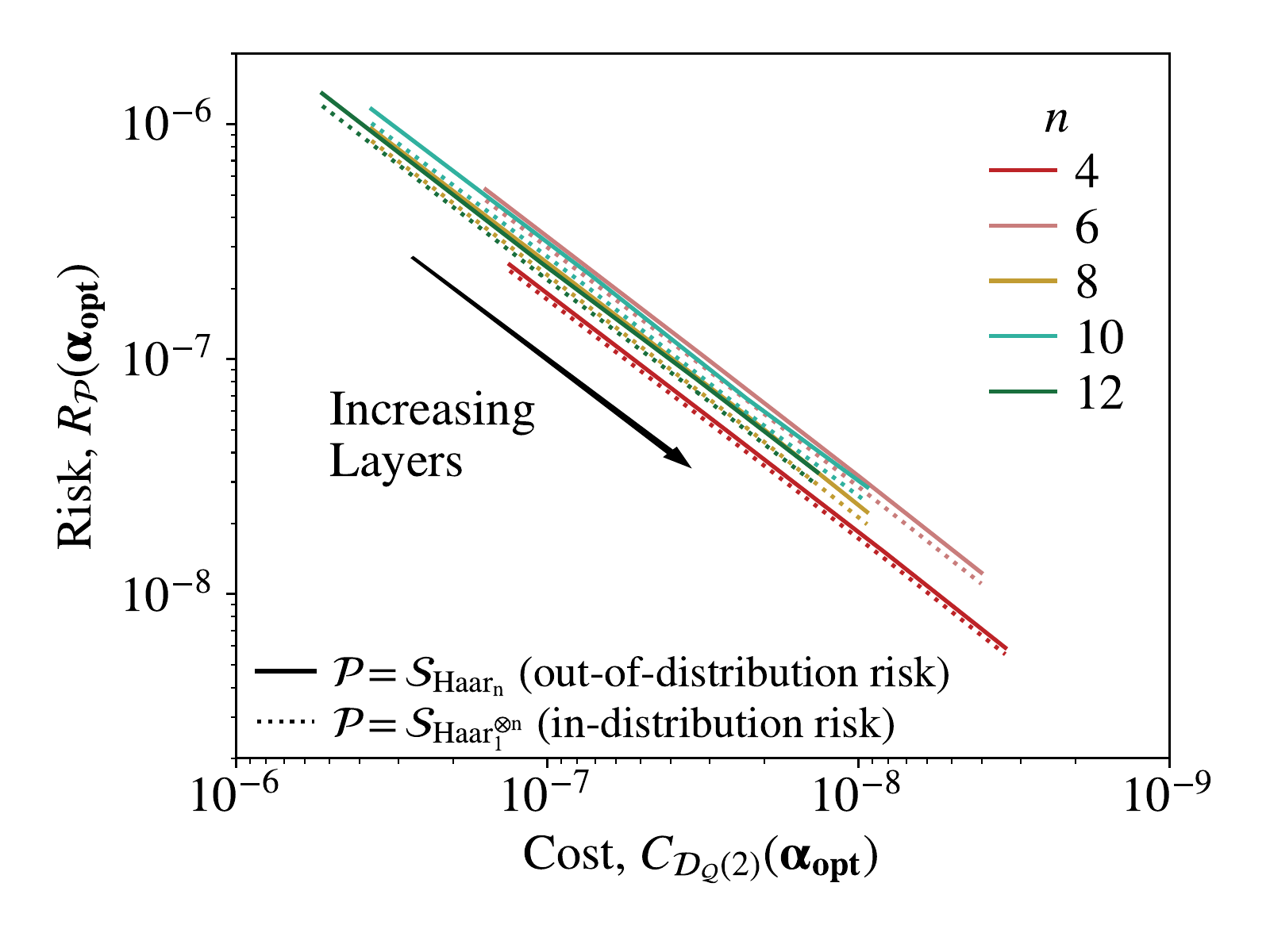}
\vspace{-6mm}
\caption{\small \textbf{Out-of-Distribution Generalization for Hamiltonian Learning.} Here we present our results from learning the Hamiltonian specified in Eq.~\eqref{eq:Hamiltonian} by training on only 2 product states. As the number of layers $L$ in the ansatz is increased the obtainable cost function value decreases. We plot the correlation between the optimized cost $C_{\mathcal{D}_{\mathcal{Q}}(2)}(\vec{\alpha_{\rm opt}})$ with $\mathcal{Q} = \mathcal{S}_{{\rm Haar}_1^{\otimes n}}$, and the (in-distribution) risk over product states, $R_{\mathcal{S}_{{\rm Haar}_1^{\otimes n}}}$, and (out-of-distribution) risk over the Haar measure, $R_{\mathcal{S}_{{\rm Haar}_n}}$. The lines indicate the joined values for $L = 2,3,4,5$ for the different values of $n$ indicated in the legend.}
\label{fig:main_numerics}
\end{figure}

\begin{figure}[t]
\centering
\includegraphics[width =\columnwidth]{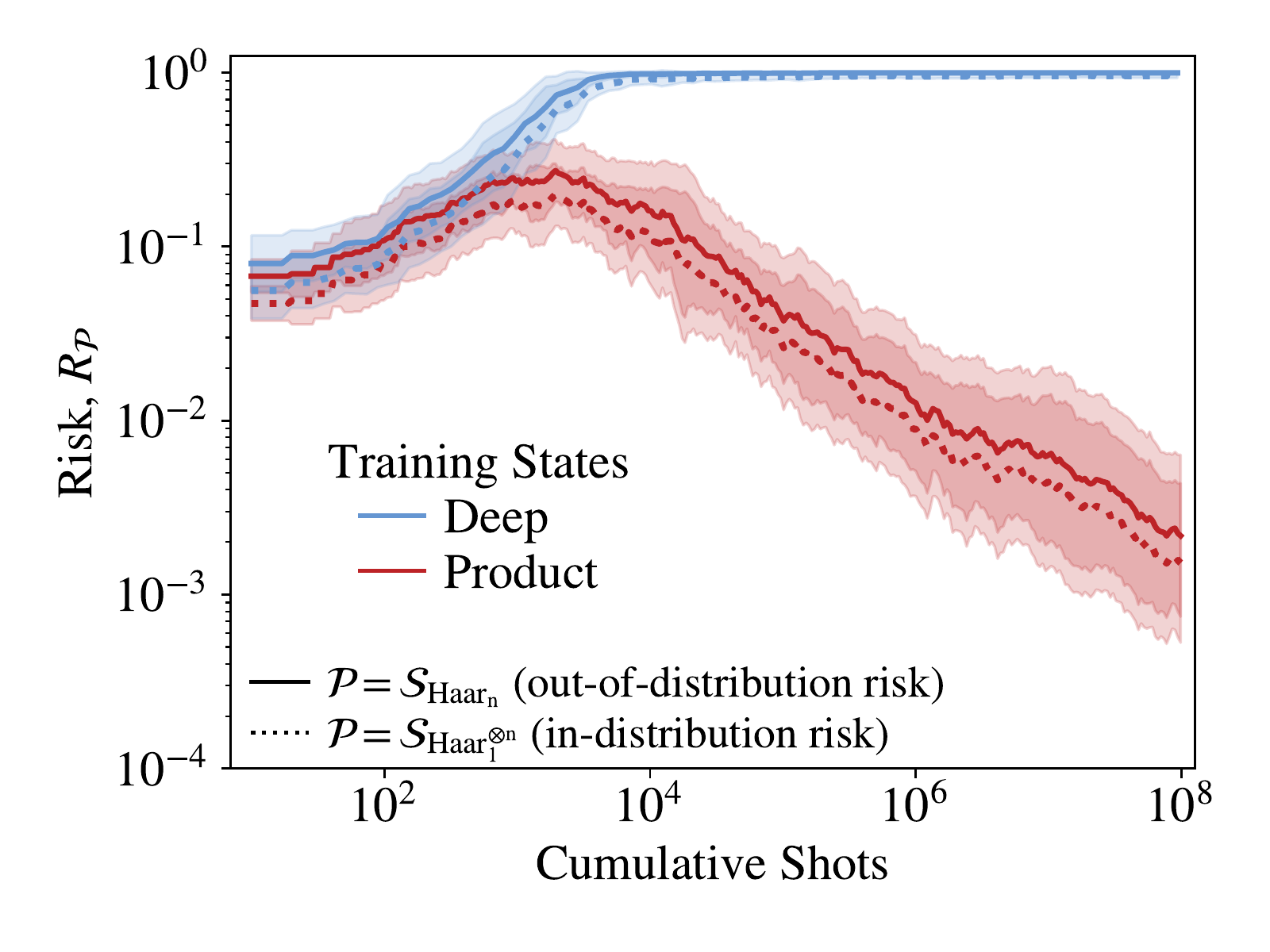}
\vspace{-6mm}
\caption{\small \textbf{Training in the presence of noise.} The cost function is optimized for two types of training data: (i) product states (red lines) and (ii) states prepared with deep circuits (blue lines). We performed 20 independent optimizations, each time initializing the optimization differently and selecting a different random training set. The shaded region represents the standard deviation of all 20 runs. Dotted (solid) lines represent in-distribution (out-of-distribution) risk.}
\label{fig:noisy_numerics}
\end{figure}

Next, we perform noisy simulations to assess the performance of learning the parameters of the Hamiltonian in Eq.~\eqref{eq:Hamiltonian} in two situations: (i) the training is performed on random product states and (ii) the training data is prepared with deep quantum circuits. We expect that the presence of noise will have a different impact depending on the amount of noise that is accumulated during the preparation of the training states.

Our simulations used a realistic noise model based on gate-set tomography on the IBM Ourense superconducting qubit device~\cite{cincio2020machine} but with the experimentally obtained error rates reduced by a factor of 20 to make the difference in training more pronounced. The training set is constructed from just two states (either product states or those prepared with a linear depth hardware efficient circuits). 

The optimizer is a version of the gradient-free Nelder-Mead method~\cite{nelder1965simplex}. The cost function in Eq.~\eqref{eq:costfid} is computed with an increasing number of shots, starting with 10 shots per cost function evaluation. That number is increased by 50\% once the optimizer detects a lack of progress within a specified number of iterations. This optimization procedure is sensitive to flatness of the cost function landscape: The flatter the landscape, the more shots are needed to resolve it and find a minimizing direction. 

Fig.~\ref{fig:noisy_numerics} shows the results of the training procedure performed on an $n=6$ qubit system. Here, we train the $L=2$ ansatz for $V_L(\vec{p},\vec{q},\vec{r})$ and consider total evolution time $t=0.1$. The optimization is repeated 20 times, each time starting with different random initial point $(\vec{p}_0,\vec{q}_0,\vec{r}_0)$. Red (blue) lines indicate the risk obtained for product (deep circuit) training states as a function of total number of shots. 

Training with product states is successful: once the number of shots per cost function evaluation is large enough (total shots above $10^3$), the optimizer detects the downhill direction and the in-distribution risk is gradually decreased, eventually reaching $10^{-3}$. The out-of-distribution risk closely follows the in-distribution risk proving that generalization can be achieved with product training states under realistic noise and finite shot budget conditions. In contrast, the training set built with deep circuits fails to produce successful training for all 20 optimization runs. Even in the limit of very large number of shots, both in-distribution and out-of-distribution risks remain large.
This proof-of-principle numerical experiment shows that our out-of-distribution generalization guarantees can make training and learning feasible in noisier scenarios than otherwise viable.

\section{Discussion}\label{sec:discussion}

Our work establishes that for learning unitaries, QNNs trained on quantum data enjoy out-of-distribution generalization between some physically relevant distributions if the training data size is roughly the number of trainable gates. The class of locally scrambled distributions that our results hold for fall naturally into sub-classes of training ensembles and testing ensembles, characterized by their practicality and generality respectively. The simplest possible training ensemble in this context are products of stabilizer states. Our results show that training on this easy to experimentally prepare and easy to classically simulate ensemble generalizes to the uniform Haar ensemble of states, as well as to practically motivated ensembles such as the output of random circuits. Thus, somewhat surprisingly, we have shown the action of quantum unitaries can be predicted on a wide class of highly entangled states, having only observed their action on relatively few unentangled states.

These results have implications for the practicality of learning quantum dynamics. We are particularly intrigued by the possibility of using quantum hardware or experimental systems to characterize unknown dynamics of quantum experimental systems. This could be done by coherently interacting a quantum system with a quantum computer, or alternatively could be conducted in a more conventional experimental setup. We stress for the latter, the experimental setup may not be equipped with a complete gate set, and so our proof that learning can be done using only products of random single qubits states, which require only simple single-qubit gates to prepare, is particularly important. 

We are also interested in the potential of these results to ease the classical compilation of local short-time evolutions into shorter depth circuits~\cite{khatri2019quantum} and circuits of a particular desired structure~\cite{cirstoiu2020variational,gibbs2021long}. 
Since low-entangling unitaries and product states may be classically simulated using tensor network methods, our results show that the compilation of such unitaries may be performed entirely classically.
This could be used to develop more effective methods for dynamical simulation or to learn more efficient pulse sequences for noise resilient gate implementations. 

An immediate extension of our results would be to investigate whether our proof techniques can be used to more efficiently evaluate Haar integrals, or more generally to relate averages over different locally scrambled ensembles in other settings. For example, one might explore whether they could be used in a DQC1 (Deterministic Quantum Computation with 1 clean qubit) setting where one inputs a maximally mixed state~\cite{knill1998power}. Alternatively, one might investigate whether they could be used to bound the frame potential of an ensemble, an important quantity for evaluating the randomness of a distribution that has links with quantifying chaotic behavior~\cite{roberts2017chaos}.

In this paper we have focused on the learning of quantum dynamics, in particular the learning of unitaries, using locally scrambled distributions. 
Given recent progress on different quantum channel learning questions \cite{flammia2020efficient, harper2020efficient, harper2021fast, flammia2021pauli, chen2022quantum, chung2021sample, caro2021binary, fanizza2022learning, huang2022foundations, huang2022learning, caro2022learning}, it is natural to ask whether out-of-distribution generalization is possible for other QML tasks such as learning quantum channels~\footnote{We note that our proof techniques extend beyond unitary dynamics to doubly stochastic quantum channels, which can be understood as affine combinations of unitary channels~\cite[Theorem 1]{mendl2009unital}.} or, more generally, for performing classification tasks such as classifying phases of matter \cite{caro2021generalization, uvarov2020machine, monaco2022quantum}. It would further be valuable to investigate whether out-of-distribution generalization is viable for other classes of distributions. Such results, if obtainable, would again have important implications for the practicality of QML on near term hardware and restricted experimental settings.

Our approach to out-of-distribution generalization does not rely on specific learning algorithms, nor transfer learning techniques, as is often the case in the classical literature~\cite{quinonero2008dataset, shimodaira2000improving, shen2021towards, pratt1991direct, pan2009survey}. Rather, we establish generalization guarantees that apply to a specific QML task (learning quantum dynamics) with data coming from a specific class of distributions (locally scrambled ensembles). That is, we show that in this context, out-of-distribution generalization is essentially automatic.
In the classical ML literature, a similar-in-spirit focus on properties of the class of distributions of interest can for example be seen in the concepts of invariance~\cite{arjovsky2019invariant, arjovsky2021out-of-distribution} and variation~\cite{ye2021towards} of features, but the nature of these properties is still quite different from the ones that we consider. Nevertheless, we hope that combining such perspectives from classical ML theory with physics-informed choices of distributions, as in our case, will lead to a better understanding of out-of-distribution generalization.

\section{Methods}

In this section, we give an overview over the 
proof strategy leading to our central analytical result contained in Lemma~\ref{lem:locally-scrambled-cost-versus-hst-cost}.
At a high level, our proof boils down to rewriting $R_{\mathcal{S}_{{\rm Haar}_n}}(\vec{\alpha})$ and $R_{\mathcal{Q}} \left( \vec{\alpha} \right)$ with $\mathcal{Q}$ locally scrambled into forms which are comparable by known and newly derived inequalities.

First, we recast the Haar risk $R_{\mathcal{S}_{{\rm Haar}_n}}(\vec{\alpha})$ into an average over Pauli products and upper bound it by a risk over local stabilizer states. 
To do so, we rewrite the Haar risk by recalling the relationship between the (Haar) average gate fidelity between two unitaries $U$ and $V$ and the Hilbert-Schmidt inner product~\cite{nielsen2002simple},  
\small
\begin{equation}
\begin{aligned}
    R_{\mathcal{S}_{{\rm Haar}_n}}(\vec{\alpha}) 
    &= \mathbb{E}_{\ket{\Psi}\sim\mathcal{S}_{{\rm Haar}_n}} \left[1 -  | \bra{\Psi}   U^\dagger V(\alv) \ket{\Psi} |^2 \right] \\
    &= \frac{d}{d+1} \left( 1 -  \frac{1}{d^2} | \Tr[U^\dagger V(\vec{\alpha}) ] |^2 \right) \, .
\end{aligned}
\end{equation}
\normalsize
Next, we use the Pauli basis expansion of the swap operator to write the Haar risk as an average over Pauli operators. That is, as shown explicitly in Lemma~\ref{lem:hst-cost-pauli-average}, we use
\begin{equation}
    \operatorname{SWAP} = \sum_{P\in\{\mathds{1}, X, Y, Z\}^{\otimes n}} P \otimes P
\end{equation}
to show that 
\begin{align}
       | \Tr[U^\dagger V ] |^2 
        &= \frac{1}{d}\sum_{P\in\{\mathds{1}, X, Y, Z\}^{\otimes n}} \Tr [ P U^\dagger V P V^\dagger U ]\, .
\end{align}
This gives an expression for the Haar risk $R_{\mathcal{S}_{{\rm Haar}_n}}(\vec{\alpha})$
in terms of an average over Pauli products. 
This average over Pauli observables can then be upper bounded by an average over products of stabilizer states by introducing a spectral decomposition, as detailed in Lemma~\ref{lem:auxiliary-overlap-bounds} and Corollary~\ref{corol:hst-cost-bound-local-haar-average}. Finally, by the 2-design property of the random single-qubit stabilizer states, we can rewrite this upper bound in terms of a local Haar average,
\begin{equation}
\begin{aligned}
    &\frac{d+1}{d}R_{\mathcal{S}_{{\rm Haar}_n}}(\vec{\alpha}) 
    \leq 2\left( 1 - \chi \right)\, \, \, \, \, \, \text{where,} \\
    &\chi = \mathbb{E}_{\bigotimes_{i=1}^n \ket{\psi_i}\sim {\rm Haar}_1^{\otimes n}} \left[ \left\lvert\left(\bigotimes_{i=1}^n \bra{\psi_i}\right) \Tilde{U}^\dagger W \Tilde{U} \left(\bigotimes_{i=1}^n \ket{\psi_i}\right) \right\rvert^2 \right] \, .
\end{aligned}
\end{equation}
The latter can then be related to $R_{\mathcal{Q}} \left( \vec{\alpha} \right)$ because $\mathcal{Q}$ is locally scrambled, which then leads to the first inequality in Lemma~\ref{lem:locally-scrambled-cost-versus-hst-cost}. Here the choice to bound by $\text{Haar}_1^{\otimes n}$ specifically hints towards our final result that a unitary can be learnt over the Haar average from product state training data.

Second, we recast the generic locally scrambled risk $R_{\mathcal{Q}}$ into a sum of locally scrambled expectation values over different partitions of the system. Specifically, using a well known expression for the complex second moment of the single-qubit Haar measure~(see, e.g., Eq. (2.26) in \cite{roberts2017chaos}), 
\begin{align}
        \mathbb{E}_{\ket{\psi}\sim{\rm Haar}_1}\left[
        \ket{\psi}\bra{\psi}^{\otimes 2} \right] 
        &= \frac{\mathds{1} \otimes \mathds{1} + \operatorname{SWAP}}{6} \, ,
\end{align}
we find that 
\small
\begin{equation}
      R_{\mathcal{Q}} \left( \vec{\alpha} \right) 
      = 1 - \frac{1}{6^n} \sum_{A \subseteq \{1, \ldots, n\}} \mathbb{E}_{\Tilde{U}\sim \tilde{\mathcal{U}}}\norm{\Tr_{A^c}\left[\Tilde{U}^\dagger U^\dagger V \Tilde{U}\right]}_F^2 \, ,
\end{equation}
\normalsize
where $\tilde{U}\sim\Tilde{\mathcal{U}}$ is drawn from the locally scrambled unitary ensemble $\Tilde{\mathcal{U}}$ with $\mathcal{Q}=\Tilde{\mathcal{U}}\ket{0}^{\otimes n}$. 
See Lemma~\ref{lem:locally-scrambled-cost-explicit-expression} for more details. From here, we can use matrix-analytic inequalities to show a lower bound on the Frobenius norm of a partial trace of a matrix in terms of the absolute value of the trace of the original matrix. Plugging this lower bound into the explicit expression for $R_{\mathcal{Q}} \left( \vec{\alpha} \right)$ translates exactly to the second inequality in Lemma~\ref{lem:locally-scrambled-cost-versus-hst-cost}.

\section*{Data availability}

The data generated and analyzed during the current study are available from the authors upon request.

\section*{Code availability}

Further implementation details are available from the authors upon request.

\vspace{0.15in}


\bibliography{quantum}

\begin{acknowledgments}
    We thank Marco Cerezo for helpful conversations.
    We thank the reviewers at Nature Communications for their valuable feedback.
    MCC was supported by the TopMath Graduate Center of the TUM Graduate School at the Technical University of Munich, Germany, the TopMath Program at the Elite Network of Bavaria, by a doctoral scholarship of the German Academic Scholarship Foundation (Studienstiftung des deutschen Volkes), by the BMWK (PlanQK), and by a DAAD PRIME Fellowship. NE was supported by the U.S. DOE, Department of Energy Computational Science Graduate Fellowship under Award Number DE-SC0020347.
    HH is supported by a Google PhD Fellowship.
    PJC and ATS acknowledge initial support from the Los Alamos National Laboratory (LANL) ASC Beyond Moore's Law project. Research presented in this paper (ATS) was supported by the Laboratory Directed Research and Development (LDRD) program of Los Alamos National Laboratory under project number 20210116DR.
    LC acknowledges support from LDRD program of LANL under project number 20230049DR. LC and PJC were also supported by the U.S. DOE, Office of Science, Office of Advanced Scientific Computing Research, under the Accelerated Research in Quantum Computing (ARQC) program. 
    ZH acknowledges support from the LANL Mark Kac Fellowship and from the Sandoz Family Foundation-Monique de Meuron program for Academic Promotion.
\end{acknowledgments}

\clearpage 
\input{supplementary_theory}

\end{document}

%% file: supplementary_theory.tex
\appendix

\setcounter{page}{1}
\renewcommand\thefigure{\thesection\arabic{figure}}
\setcounter{figure}{0} 

\onecolumngrid

\begin{center}
\large{ Supplementary Material for \\ ``Out-of-Distribution Generalization for Learning Quantum Dynamics''
}
\end{center}

\section{Preliminaries}\label{app:Prelim}

\medskip 

\setcounter{thm}{0}
\setcounter{corol}{0}
\setcounter{lem}{0}
\setcounter{propos}{0}
\setcounter{defn}{0}
\setcounter{rmk}{0}
\setcounter{ex}{0}

Before beginning our main discussion, we introduce some notation that will be used throughout the Supplementary Material.
In our discussion, we consider systems consisting of $n$ qubits. Thus, we work with the complex Hilbert space $(\mathbb{C}^2)^{\otimes n}$ of dimension $d=2^n$.
For any $d$, $\mathcal{B}(\mathbb{C}^d)$ denotes the set of bounded linear operators on $\mathbb{C}^d$, which we implicitly identify with the set of $d\times d$ matrices by fixing a basis whenever convenient.
Also, we denote by $\mathcal{U}(\mathbb{C}^d)$ the set of unitary operators on $\mathbb{C}^d$. The sets of operators $\mathcal{B}((\mathbb{C}^2)^{\otimes n})$ and $\mathcal{U}((\mathbb{C}^2)^{\otimes n})$ are defined analogously.
Finally, we use standard bra-ket notation for pure quantum states.

We consider the task of learning an unknown $n$-qubit unitary $U\in\mathcal{U}((\mathbb{C}^{2})^{\otimes n})$ from pairs of input and output states using a \emph{quantum neural network (QNN)}.
For our purposes, we think of a QNN as a (possibly variable-structure) $k$-local quantum circuit on $n$ qubits that contains tunable gates. (Here, $k$ is an $n$-independent constant). 
Mathematically, we describe such a QNN by a parameterized $n$-qubit unitary $V(\alv)$ with classical parameters $\alv$, where the parameterization arises from the QNN structure.
The parameter vector $\alv$ can consist of both continuous parameters (which indeed parameterize the trainable gates, e.g. acting as rotation angles) and discrete parameters (which encode freedom in the chosen quantum circuit structure, e.g. the number of trainable gates). 
The input to the quantum learning procedure is a training data set $\mathcal{D}_{\mathcal{Q}} (N)$ of the form 
\begin{equation}\label{eq:supplementary-training-data-general}
    \mathcal{D}_{\mathcal{Q}} (N) = \{(\ket{\Psi^{(j)}},\ket{\Phi^{(j)}}) \}_{j=1}^{N}\, 
\end{equation}
where the $\ket{\Psi^{(j)}}\in (\mathbb{C}^{2})^{\otimes n}$ are pure $n$-qubit input states drawn i.i.d.~from a probability distribution $\mathcal{Q}$ and $\ket{\Phi^{(j)}} = U \ket{\Psi^{(j)}}$ are the corresponding output states.
The goal is to train the classical parameters $\alv$ in the QNN $V(\alv)$ such that the QNN $V(\alv_{{\rm opt}})$ with the optimized parameters $\alv_{{\rm opt}}$ predicts the output states of $U$ well on average when the input states are drawn from a testing probability distribution $\mathcal{P}$ over pure $n$-qubit states.
That is, the optimized parameters $\alv_{{\rm opt}}$ should be such that 
\begin{align}
    R_{\mathcal{P}} \left( U, V(\alv_{{\rm opt}}) \right)
    &\coloneqq \frac{1}{4}\mathbb{E}_{\ket{\Psi}\sim \mathcal{P}} \left[ \norm{U\ket{\Psi}\bra{\Psi}U^\dagger - V(\alv_{{\rm opt}})\ket{\Psi}\bra{\Psi}V(\alv_{{\rm opt}})^\dagger}_1^2\right]\\ 
    &= 1 - \mathbb{E}_{\ket{\Psi}\sim \mathcal{P}} \left[ \left\lvert \bra{\Psi} U^\dagger V(\alv_{{\rm opt}}) \ket{\Psi} \right\rvert^2 \right]\label{eq:supplementary-test-cost-general}
\end{align}
is small. 
A learner who does not know the testing distribution $\mathcal{P}$ and the target unitary $U$ cannot evaluate the expected testing risk from Eq.~\eqref{eq:supplementary-test-cost-general}. 
Instead, given a training data set as in Eq.~\eqref{eq:supplementary-training-data-general}, the learner may try to evaluate and optimize the training cost 
\begin{align}
    C_{\mathcal{D}_{\mathcal{Q}}(N)} (U, V(\alv))
    &\coloneqq \frac{1}{4N} \sum_{j=1}^N \norm{U\ket{\Psi^{(j)}}\bra{\Psi^{(j)}}U^\dagger - V(\alv)\ket{\Psi^{(j)}}\bra{\Psi^{(j)}}V(\alv)^\dagger}_1^2\\
    &= 1 - \frac{1}{N} \sum_{j=1}^N \left\lvert \bra{\Psi^{(j)}} U^\dagger V(\alv)\ket{\Psi^{(j)}}\right\vert^2 \, .\label{eq:supplementary-training-cost-general}
\end{align}
Here, we rewrite the trace norm distance between two pure states in terms of their fidelity to obtain an expression for the training cost that can be evaluated on a quantum computer with a swap test~\cite{buhrman2001quantum, gottesman2001quantum}.

Optimizing the training cost from Eq.~\eqref{eq:supplementary-training-cost-general}, however, is not automatically a promising avenue towards achieving a small expected testing risk from Eq.~\eqref{eq:supplementary-test-cost-general}.
Such a promise can only be fulfilled when a good performance on the available training data, i.e. a small value of $C_{\mathcal{D}_{\mathcal{Q}}(N)} (U, V(\alv_{{\rm opt}}))$, also leads to a good average performance on previously unseen data points, i.e. to a small value of $R_{\mathcal{P}} \left( U, V(\alv_{{\rm opt}}) \right)$.
This ability to \emph{generalize} from training data to unseen data is of central importance to the viability of (quantum) machine learning. In particular, the generalization behavior often has a determining influence on the amount of training data that a (quantum) machine learning model requires.

For the case of $\mathcal{Q} = \mathcal{P}$, when training and testing data are drawn i.i.d.~from the same distribution, such questions can be studied in the standard framework of \emph{in-distribution generalization} (sometimes also known as \emph{weak generalization}). 
Here, however, we focus on the case of $\mathcal{Q} \neq  \mathcal{P}$, when the QNN is trained on a distribution different from the testing distribution. This scenario is variously known as \textit{out-of-distribution generalization} or \textit{strong generalization}. 
More precisely, we will consider training and testing states coming from (different) locally scrambled ensembles~\cite{kuo2020markovian, hu2021classical}.

\begin{defn}[Locally scrambled ensembles -- Restatement of Definition~\ref{definition:locally-scrambled-ensembles}]
    An ensemble $\mathcal{U}$ of $n$-qubit unitaries is called \emph{locally scrambled} if it is invariant under preprocessing by tensor products of arbitrary local unitaries. 
    That is, if $U\sim\mathcal{U}$, then for any fixed $U_1,\ldots, U_n\in\mathcal{U}(\mathbb{C}^2)$ also $U (\bigotimes_{i=1}^n U_i)\sim\mathcal{U}$.
    Accordingly, an ensemble $\mathcal{S}$ of $n$-qubit quantum states is called locally scrambled if it is of the form $\mathcal{S}=\mathcal{U}\ket{0}^{\otimes n}$ for some locally scrambled ensemble $\mathcal{U}$ of $n$-qubit unitaries.
    
    We use $\mathbb{U}_{\rm LS}$ to denote the class of all locally scrambled unitary ensembles and $\mathbb{S}_{\rm LS}$ to denote the class of all locally scrambled state ensembles. (Here, we suppress the number of qubits in the notation in favor of improved readability.)
\end{defn}

As discussed in more detail below, examples of locally scrambled unitary ensembles include Haar-random $n$-qubit unitaries, tensor products of Haar-random single-qubit unitaries followed by some fixed $n$-qubit unitary, and unitaries implemented by random quantum circuits of some fixed depth, among others.

In fact, our results for locally scrambled ensembles below immediately extend to a slightly broader class of ensembles:

\begin{defn}
    An ensemble $\mathcal{U}$ of $n$-qubit unitaries is called \emph{locally scrambled up to (and including) complex second moments} if there exists an ensemble $\Tilde{\mathcal{U}}$ of $n$-qubit unitaries such that the complex first and second moments of $\mathcal{U}$ agree with those of $\Tilde{\mathcal{U}}$.
    That is, we have $\mathbb{E}_{U\sim\mathcal{U}}[U\rho U^\dagger] = \mathbb{E}_{\tilde{U}\sim\tilde{\mathcal{U}}}[\tilde{U}\rho \tilde{U}^\dagger]$ for all density matrices $\rho\in\mathcal{B}(\mathbb{C}^d)$ and $\mathbb{E}_{U\sim\mathcal{U}}[U^{\otimes 2}\rho (U^\dagger)^{\otimes 2}] = \mathbb{E}_{\tilde{U}\sim\tilde{\mathcal{U}}}[\tilde{U}^{\otimes 2}\rho (\tilde{U}^\dagger)^{\otimes 2}]$ for all density matrices $\rho\in\mathcal{B}((\mathbb{C}^d)^{\otimes 2})$.
    Accordingly, an ensemble $\mathcal{S}$ of $n$-qubit quantum states is called locally scrambled up to (and including) complex second moments if it is of the form $\mathcal{S}=\mathcal{U}\ket{0}^{\otimes n}$ for ensemble $\mathcal{U}$ of $n$-qubit unitaries that is locally scrambled up to (and including) complex second moments.
    
    We use $\mathbb{U}_{\rm LS}^{(2)}$ to denote the class of all unitary ensembles that are locally scrambled up to (and including) complex second moments and $\mathbb{S}_{\rm LS}^{(2)}$ to denote the class of all state ensembles that are locally scrambled up to (and including) complex second moments. (Again, we suppress the number of qubits in the notation in favor of improved readability.)
\end{defn}

To illustrate the above definitions here we list examples of distribution in $\mathbb{S}_{\rm LS}$ and $\mathbb{S}_{\rm LS}^{(2)}$ and explain why they are such distributions.

\begin{ex}[Products of Haar-random single-qubit states, $\mathcal{S}_{{\rm Haar}_1^{\otimes n}}$]\label{ex:product-haar-random}
    Define $\mathcal{S}_{{\rm Haar}_1^{\otimes n}} \coloneqq \mathcal{U}_{{\rm Haar}_1^{\otimes n}}|0\rangle^{\otimes n}$, where $\mathcal{U}_{{\rm Haar}_1^{\otimes n}}$ is the $n$-fold tensor product of the single-qubit Haar measure on $\mathcal{U}(\mathbb{C}^2)$.
    By the right-invariance of $\mathcal{U}_{{\rm Haar}_1}$ under multiplication with an arbitrary fixed single-qubit unitary we have that $\mathcal{U}_{{\rm Haar}_1^{\otimes n}} \in \mathbb{U}_{\rm LS}$ and hence $\mathcal{S}_{{\rm Haar}_1^{\otimes n}} \in \mathbb{S}_{\rm LS}$.
\end{ex}

\begin{ex}[Products of random single-qubit stabilizer states, $\mathcal{S}_{{\rm Stab}_1^{\otimes n}}$]\label{ex:product-random-stabilizer}
Consider the ensemble composed of tensor products of random single-qubit stabilizer states, i.e. $\mathcal{S}_{{\rm Stab}_1^{\otimes n}} := \operatorname{Uniform} (\{\ket{0},\ket{1},\ket{+},\ket{-}, \ket{y+},\ket{y-}\}^{\otimes n})$. This ensemble agrees with $\mathcal{S}_{{\rm Haar}_1^{\otimes n}}$ up to the second moment and hence is in $\mathbb{S}_{\rm LS}^{(2)}$.
\end{ex}

\begin{ex}[Post-processed tensor products of Haar-random single-qubit states]\label{ex:processed-product-haar-random}
    Consider an ensemble of $n$-qubit unitaries of the form $U \mathcal{U}_{{\rm Haar}_1^{\otimes n}}$, where $\mathcal{U}_{{\rm Haar}_1^{\otimes n}}$ denotes the $n$-fold tensor product of the single-qubit Haar measure on $\mathcal{U}(\mathbb{C}^2)$ and $U\in\mathcal{U}((\mathbb{C}^{2})^{\otimes n})$ is some fixed unitary.
    Then, by right-invariance of $\mathcal{U}_{{\rm Haar}_1}$ under multiplication with an arbitrary fixed single-qubit unitary, we have $U \mathcal{U}_{{\rm Haar}_1^{\otimes n}} \left(\bigotimes_{i=1}^n U_i\right) = U \mathcal{U}_{{\rm Haar}_1^{\otimes n}}$, for any fixed unitaries $U_1,\ldots,U_n\in \mathcal{U}(\mathbb{C}^2)$. 
    Thus, $U \mathcal{U}_{{\rm Haar}_1^{\otimes n}}$ is a locally scrambled ensemble of $n$-qubit unitaries and $U \mathcal{U}_{{\rm Haar}_1^{\otimes n}}\ket{0}^{\otimes n}\in \mathbb{S}_{\rm LS}$. 
\end{ex}

\begin{ex}[Haar-random $n$-qubit states, $\mathcal{S}_{{\rm Haar}_n}$]\label{ex:global-haar-random}
    Haar-random $n$-qubit states, i.e. $\mathcal{S}_{\rm Haar}:= \mathcal{U}_{{\rm Haar}_n}|0\rangle^{\otimes n}$, form a locally scrambled ensemble because the Haar measure on $\mathcal{U}((\mathbb{C}^{2})^{\otimes n})$ is (both left- and) right-invariant under multiplication with an arbitrary fixed $n$-qubit unitary.
    So, $\mathcal{U}_{{\rm Haar}_n}\in\mathbb{U}_{\rm LS}$ and $\mathcal{S}_{\rm Haar}\in \mathbb{S}_{\rm LS}$.
\end{ex}

\begin{ex}[$2$-design on $n$-qubit states, $\mathcal{S}_{\rm 2design}$]\label{ex:2-design}
A unitary $2$-design, by definition, agrees with the Haar distribution up to the second moment and hence $\mathcal{S}_{\rm 2design} := \mathcal{U}_{\rm 2} |0 \rangle^{\otimes n} \in \mathbb{S}_{\rm LS}^{(2)}$. Such distributions can be well approximated using a polynomial-depth hardware-efficient ansatz~\cite{brandao2016local, harrow2018approximate, haferkamp2022randomquantum}. 
\end{ex}

\begin{ex}[Products of Haar-random $k$-qubit states, $\mathcal{S}_{{\rm Haar}_k^{\otimes n/k}}$]
    For $n/k \in \mathbb{N}$ we can capture both $\mathcal{S}_{{\rm Haar}_n}$ and $\mathcal{S}_{{\rm Haar}_1^{\otimes n}}$ by more generally defining $\mathcal{S}_{{\rm Haar}_k^{\otimes n/k}} \coloneqq \mathcal{U}_{{\rm Haar}_k^{\otimes n/k}}|0\rangle^{\otimes n}$, where $\mathcal{U}_{{\rm Haar}_1^{\otimes n/k}}$ is the $n/k$-fold tensor product of the $k$-qubit Haar measure on $\mathcal{U}((\mathbb{C}^{2})^{\otimes k})$.
    Again, right-invariance of the Haar measure yields $\mathcal{S}_{{\rm Haar}_k^{\otimes n/k}}\in \mathbb{S}_{\rm LS}$.
\end{ex}

\begin{ex}[Output states of random quantum circuits, $\mathcal{S}_{\rm RandCirc}^{\mathcal{A}_k}$]\label{ex:random-circuit-outputs}
    Let $\mathcal{A}_k$ be a $k$-local $n$-qubit quantum circuit architecture (in which every qubit is acted on non-trivially), with $k\leq n$. 
    Let $\mathcal{U}_{\mathcal{A}_k}$ denote the ensemble of $n$-qubit unitaries obtained by drawing every $k$-qubit unitary in $\mathcal{A}_k$ at random from the $k$-qubit Haar measure. 
    Then, by right-invariance of the $k$-qubit Haar measure, $\mathcal{U}_{\mathcal{A}_k}\in \mathbb{U}_{\rm LS}$.
    Accordingly, the ensemble $\mathcal{S}_{\rm RandCirc} \coloneqq \mathcal{U}_{\mathcal{A}_k} \ket{0}^{\otimes n}$ of output states of a random quantum circuit with architecture $\mathcal{A}_k$ satisfies $\mathcal{S}_{\rm RandCirc}^{\mathcal{A}_k}\in \mathbb{S}_{\rm LS}$.
\end{ex}

\section{Analytical Results}\label{app:AnalyticResults}

\subsection{Equivalence of Locally Scrambled Risks for Comparing Unitaries}\label{subsection:supplementary-equivalence-locally-scrambled-costs}

We begin our analysis by comparing the testing risks obtained from Eq.~\eqref{eq:supplementary-test-cost-general} for different locally scrambled testing distributions.
In this subsection, we show that all such testing risks are equivalent in the sense that they differ by at most a constant factor.
We prove this equivalence by showing that all locally scrambled risks $R_{\mathcal{P}} \left( U, V \right)$ are tightly related to the Hilbert-Schmidt inner product between $U$ and $V$.
To formalize this discussion, we first introduce a cost arising naturally from that inner product:

\begin{defn}[Hilbert-Schmidt test cost]\label{dfn:hst-cost}
    The \emph{Hilbert-Schmidt test (HST) cost} between two $n$-qubit unitaries $U\in\mathcal{U}((\mathbb{C}^{2})^{\otimes n})$ and $V\in\mathcal{U}((\mathbb{C}^{2})^{\otimes n})$ is defined as 
    \begin{equation}
        C_{\mbox {\tiny \rm HST}}(U, V)
        \coloneqq 1-\frac{1}{d^{2}}|\Tr[ U^\dagger V]|^2 \, . \label{eq:hst-cost}
    \end{equation}
\end{defn}

At this point, we note that, as shown in Refs.~\cite{khatri2019quantum, nielsen2002simple}, we can view the HST cost as an expected testing risk as in Eq.~\eqref{eq:supplementary-test-cost-general} via
\begin{equation}\label{eq:hst-cost-versus-expected-haar-cost}
    C_{\mbox {\tiny \rm HST}}(U, V)
    = \frac{d+1}{d} R_{\mathcal{S}_{\rm Haar_n}}(U,V)\, ,
\end{equation}
where $d=2^n$. (Here, the notation $R_{\mathcal{S}_{{\rm Haar}_n}}(U,V)$ indicates an expected testing risk w.r.t.~the Haar ensemble from Example~\ref{ex:global-haar-random}.)
In the main text, for conciseness of presentation, we have used the right hand-side of Eq.~\eqref{eq:hst-cost-versus-expected-haar-cost} instead of its left-hand side.

Our first result is an expression for the squared absolute value of the Hilbert-Schmidt inner product between two matrices -- which in particular gives an expression for the HST cost between two unitaries -- in terms of an average over $n$-qubit Paulis:

\begin{lem}\label{lem:hst-cost-pauli-average}
    Let $n\in\mathbb{N}$ and write $d=2^n$.
    Let $A,B\in \mathcal{B}(\mathbb{C}^{d})$. Then, 
    \begin{equation}
        \lvert \Tr [A^\dagger B]\rvert^2
        = \frac{1}{d} \sum_{P\in\{\mathds{1}, X, Y, Z\}^{\otimes n}} \Tr \left[ P A^\dagger B P B^\dagger A \right] \, .\label{eq:hs-inner-product-pauli-average}
    \end{equation}
    In particular, we can express the HST cost between two $n$-qubit unitaries $U\in\mathcal{U}((\mathbb{C}^{2})^{\otimes n})$ and $V\in\mathcal{U}((\mathbb{C}^{2})^{\otimes n})$ as
    \begin{equation}
        C_{\mbox {\tiny \rm HST}}(U, V)
        = 1 - \frac{1}{d^3} \sum_{P\in\{\mathds{1}, X, Y, Z\}^{\otimes n}} \Tr \left[ P U^\dagger V P V^\dagger U \right] \, .\label{eq:hst-cost-pauli-average}
    \end{equation}
\end{lem}
\begin{proof}
    Using the shorthand $C=A^\dagger B\in \mathcal{B}(\mathbb{C}^{d})$, we have
    \begin{align}
        \lvert \Tr [A^\dagger B]\rvert^2
        &= \lvert \Tr [ C ]\rvert^2\\
        &= \Tr [C] \overline{\Tr [C]}\\
        &= \Tr [C] \Tr [C^\dagger]\\
        &= \Tr [C \otimes C^\dagger ]\\
        &= \Tr [ \operatorname{SWAP}^2 (C \otimes C^\dagger) ]\\
        &= \frac{1}{d}\sum_{P\in\{\mathds{1}, X, Y, Z\}^{\otimes n}} \Tr [ \operatorname{SWAP} (PC \otimes PC^\dagger) ]\\
        &= \frac{1}{d} \sum_{P\in\{\mathds{1}, X, Y, Z\}^{\otimes n}} \Tr [ PC PC^\dagger ]\, .
    \end{align}
    Here, the first line is due to our shorthand, the second line is $\lvert z\rvert^2 = z\overline{z}$ for $z\in\mathbb{C}$, the third line uses $\Tr [C^\dagger] = \overline{\Tr [C]}$, the fourth line uses $\Tr [D \otimes E] = \Tr [D] \otimes \Tr [E]$, the fifth line uses $\operatorname{SWAP}^2 = \mathds{1}$, the sixth line uses the basis expansion $\operatorname{SWAP} = \frac{1}{d} \sum_{P\in\{\mathds{1}, X, Y, Z\}^{\otimes n}} P\otimes P$ of the swap operator in the Pauli string basis, and the last line uses $\Tr [\operatorname{SWAP} (D\otimes E)]=\Tr [DE]$. 
    This establishes Eq.~\eqref{eq:hs-inner-product-pauli-average}.
    Now, plugging Eq.~\eqref{eq:hs-inner-product-pauli-average} into the definition of the HST cost gives Eq.~\eqref{eq:hst-cost-pauli-average}.
\end{proof}

Next, we present a technical lemma which we later use to control the Pauli average in the expression for the HST cost between two unitaries:

\begin{lem}\label{lem:auxiliary-overlap-bounds}
    Let $n\in\mathbb{N}$ and write $d=2^n$.
    Let $W\in\mathcal{U}((\mathbb{C}^{2})^{\otimes n})$, let $P\in \{\mathds{1}, X, Y, Z\}^{\otimes n}$, and let $\ket{s}\in \{\ket{0},\ket{1},\ket{+},\ket{-}, \ket{y+},\ket{y-}\}^{\otimes n}$ be an eigenvector of $P$.
    Then, 
    \begin{equation}
        0
        \leq 1 - \bra{s}P\ket{s}\cdot \bra{s} W^\dagger P W\ket{s}
        \leq 2\left( 1 - \lvert\bra{s}W\ket{s}\rvert^2\right)\, .
    \end{equation}
\end{lem}
\begin{proof}
    Let $\ket{s_1}=\ket{s},\ket{s_2},\ldots,\ket{s_{d}}$ be an orthonormal basis consisting of eigenvectors of $P$. Then, by plugging in the spectral decomposition $P = \sum_{i=1}^{d} \bra{s_i}P\ket{s_i}\cdot\ket{s_i}\bra{s_i}$, we get, using that the eigenvalues of $P$ lie in $\{-1,1\}$,
    \begin{align}
        \bra{s}P\ket{s}\cdot \bra{s} W^\dagger P W\ket{s}
        &= \sum_{i=1}^d \bra{s}P\ket{s}\cdot \bra{s_i}P\ket{s_i}\cdot \bra{s}W^\dagger \ket{s_i}\bra{s_i}W\ket{s}\\
        &= \bra{s}P\ket{s}^2\cdot \bra{s} W^\dagger \ket{s}\bra{s} W\ket{s} + \sum_{i=2}^d \bra{s}P\ket{s}\cdot \bra{s_i}P\ket{s_i}\cdot \bra{s}W^\dagger \ket{s_i}\bra{s_i}W\ket{s}\\
        &= \underbrace{\lvert \bra{s} W\ket{s}\rvert^2}_{=: p_1} + \sum_{i=2}^d \underbrace{\bra{s}P\ket{s}\cdot \bra{s_i}P\ket{s_i}}_{\in\{-1,1\}}\cdot \underbrace{\bra{s}W^\dagger \ket{s_i}\bra{s_i}W\ket{s}}_{=: p_i} \label{eq:proof-overlap-bounds}\, .
    \end{align}
    With this notation, we have $0\leq p_i\leq 1$ for all $1\leq i\leq d$, where the upper bound holds by Cauchy-Schwarz and unitarity of $W$, as well as $\sum_{i=1}^d p_i = \bra{s}W^\dagger W\ket{s} = \braket{s|s}=1$, by unitarity of $W$. 
    Thus, $\{p_i\}_{i=1}^d$ is a probability vector of length $d$.
    Therefore, from Eq.~\eqref{eq:proof-overlap-bounds}, we conclude
    \begin{equation}
        \bra{s}P\ket{s}\cdot \bra{s} W^\dagger P W\ket{s}\in [p_1 - (1-p_1), p_1 + (1-p_1)]=[2p_1-1,1]\, .
    \end{equation}
    Accordingly, we obtain
    \begin{equation}
        1 - \bra{s}P\ket{s}\cdot \bra{s} W^\dagger P W\ket{s} \in [0,1-(2p_1-1)] = [0,2(1-p_1)]\, ,
    \end{equation}
    as claimed.
\end{proof}

We emphasize that this proof, in contrast to those of Lemma~\ref{lem:hst-cost-pauli-average} (and also Lemma~\ref{lem:locally-scrambled-cost-explicit-expression} below), explicitly uses the unitarity of the matrix $W$. 
This is also why we assume unitarity of $U$ and $V$ in Lemma~\ref{lem:supplementary-locally-scrambled-cost-versus-hst-cost} below, since we will again consider $W=U^\dagger V$.

Lemmas~\ref{lem:hst-cost-pauli-average} and~\ref{lem:auxiliary-overlap-bounds} allow us to prove the following upper bound on the HST cost in terms of an average over tensor products of Haar-random single-qubit states:

\begin{corol}\label{corol:hst-cost-bound-local-haar-average}
    Let $n\in\mathbb{N}$ and write $d=2^n$.
    Let $U,V,\tilde{U}\in\mathcal{U}((\mathbb{C}^{2})^{\otimes n})$.
    Then,
    \begin{equation}
        C_{\mbox {\tiny \rm HST}}(U, V)
        \leq  2\left( 1 - \mathbb{E}_{\bigotimes_{i=1}^n \ket{\psi_i}\sim {\rm Haar}_1^{\otimes n}} \left[ \left\lvert\left(\bigotimes_{i=1}^n \bra{\psi_i}\right) \Tilde{U}^\dagger W \Tilde{U} \left(\bigotimes_{i=1}^n \ket{\psi_i}\right) \right\rvert^2 \right]\right)\, ,
    \end{equation}
    where we again use the shorthand $W=U^\dagger V$.
\end{corol}
\begin{proof}
    We begin with the expression for the HST cost in terms of a Pauli average derived in Lemma~\ref{lem:hst-cost-pauli-average}. 
    For any fixed $\Tilde{U}\in \mathcal{U}((\mathbb{C}^{2})^{\otimes n})$, the definition of $C_{\mbox {\tiny \rm HST}}(U, V )$ and  Lemma~\ref{lem:hst-cost-pauli-average} imply:
    \begin{equation}
        C_{\mbox {\tiny \rm HST}}(U, V )
        = C_{\mbox {\tiny \rm HST}}(\Tilde{U}^\dagger U \Tilde{U}, \Tilde{U}^\dagger V \Tilde{U})
        = 1 - \frac{1}{d^3} \sum_{P\in\{\mathds{1}, X, Y, Z\}^{\otimes n}} \Tr \left[ P \Tilde{U}^\dagger W \Tilde{U} P  \Tilde{U}^\dagger W^\dagger \Tilde{U} \right] \, .
    \end{equation}
    Now, we can consider a spectral decomposition for $P\in\{\mathds{1}, X, Y, Z\}^{\otimes n}$, which -- since we are dealing with Pauli strings -- we can write as follows:
    \begin{equation}\label{eq:pauli-spectral-decomposition}
        P
        = \sum_{\ket{s}\in\{\ket{0},\ket{1},\ket{+},\ket{-}, \ket{y+},\ket{y-}\}^{\otimes n}:\bra{s}P\ket{s}\neq 0} \bra{s}P\ket{s}\cdot\ket{s}\bra{s}\, .
    \end{equation}
    Plugging this spectral decomposition into Eq.~\eqref{eq:hst-cost-pauli-average} to evaluate the trace, we obtain
    \begin{align}
        C_{\mbox {\tiny \rm HST}}(U, V)
        &= 1 - \frac{1}{d^3} \sum_{P\in\{\mathds{1}, X, Y, Z\}^{\otimes n}} \Tr \left[ P \Tilde{U}^\dagger W \Tilde{U} P \Tilde{U}^\dagger W^\dagger \Tilde{U} \right]\\
        &= 1 - \frac{1}{d^2} \sum_{P\in\{\mathds{1}, X, Y, Z\}^{\otimes n}} \frac{1}{d}\sum_{\ket{s}\in\{\ket{0},\ket{1},\ket{+},\ket{-}, \ket{y+},\ket{y-}\}^{\otimes n}:\bra{s}P\ket{s}\neq 0} \Tr\left[ P \Tilde{U}^\dagger W \Tilde{U} \left(\bra{s}P\ket{s}\cdot\ket{s}\bra{s}\right) \Tilde{U}^\dagger W^\dagger \Tilde{U} \right]\\
        &= 1 - \frac{1}{d^2} \sum_{P\in\{\mathds{1}, X, Y, Z\}^{\otimes n}} \frac{1}{d}\sum_{\ket{s}\in\{\ket{0},\ket{1},\ket{+},\ket{-}, \ket{y+},\ket{y-}\}^{\otimes n}:\bra{s}P\ket{s}\neq 0} \bra{s}P\ket{s} \bra{s} \Tilde{U}^\dagger W^\dagger \Tilde{U} P \Tilde{U}^\dagger W \Tilde{U}\ket{s}\\
        &= \mathbb{E}_{P\sim\{\mathds{1}, X, Y, Z\}^{\otimes n}} \mathbb{E}_{\ket{s}\sim\{\ket{0},\ket{1},\ket{+},\ket{-}, \ket{y+},\ket{y-}\}^{\otimes n}:\bra{s}P\ket{s}\neq 0}\left[ 1 - \bra{s}P\ket{s}\cdot \bra{s} \Tilde{U}^\dagger W^\dagger \Tilde{U} P \Tilde{U}^\dagger W \Tilde{U} \ket{s} \right] \\
        &= \mathbb{E}_{\ket{s}\sim\{\ket{0},\ket{1},\ket{+},\ket{-}, \ket{y+},\ket{y-}\}^{\otimes n}}\mathbb{E}_{P\sim\{\mathds{1}, X, Y, Z\}^{\otimes n}:\bra{s}P\ket{s}\neq 0}\left[ 1 - \bra{s}P\ket{s}\cdot \bra{s} \Tilde{U}^\dagger W^\dagger \Tilde{U} P \Tilde{U}^\dagger W \Tilde{U} \ket{s} \right] \, .
    \end{align}
    Here, we denote by $\mathbb{E}_{P\sim\{\mathds{1}, X, Y, Z\}^{\otimes n}}$ the expectation over uniformly random Pauli strings of length $n$, and $\mathbb{E}_{\ket{s}\sim\{\ket{0},\ket{1},\ket{+},\ket{-}, \ket{y+},\ket{y-}\}^{\otimes n}:\bra{s}P\ket{s}\neq 0}$ denotes the expectation over uniformly random tensor products of single-qubit stabilizer states which have non-vanishing overlap with $P$. Equivalently, the latter is the expectation over uniformly random eigenvectors of $P$. 
    Similarly, $\mathbb{E}_{\ket{s}\sim\{\ket{0},\ket{1},\ket{+},\ket{-}, \ket{y+},\ket{y-}\}^{\otimes n}}$ denotes the expectation over uniformly random tensor products of single-qubit stabilizer states, and $\mathbb{E}_{P\sim\{\mathds{1}, X, Y, Z\}^{\otimes n}:\bra{s}P\ket{s}\neq 0}$ denotes the expectation over uniformly random Pauli strings of length $n$ that have non-vanishing overlap with $\ket{s}$. Equivalently, the latter is the expectation over uniformly random Pauli strings that have $\ket{s}$ as an eigenvector.
    Note that the expectation values involved here are w.r.t.~uniform distributions over finite sets, which in particular justifies the last step in the above computation (since this then becomes a mere reordering of a finite sum).
    Plugging in the upper bound of Lemma~\ref{lem:auxiliary-overlap-bounds}, applied for the unitary $\Tilde{U}^\dagger W \Tilde{U}$, we further obtain:
    \begin{align}
        C_{\mbox {\tiny \rm HST}}(U, V)
        &\leq 2\mathbb{E}_{\ket{s}\sim\{\ket{0},\ket{1},\ket{+},\ket{-}, \ket{y+},\ket{y-}\}^{\otimes n}}\mathbb{E}_{P\sim\{\mathds{1}, X, Y, Z\}^{\otimes n}:\bra{s}P\ket{s}\neq 0} \left[ 1 - \lvert\bra{s} \Tilde{U}^\dagger W \Tilde{U} \ket{s}\rvert^2\right]\\
        &= 2\left( 1 - \mathbb{E}_{\ket{s}\sim\{\ket{0},\ket{1},\ket{+},\ket{-}, \ket{y+},\ket{y-}\}^{\otimes n}} \left[ \lvert\bra{s} \Tilde{U}^\dagger W \Tilde{U} \ket{s}\rvert^2 \right]\right)\\
        &= 2\left( 1 - \mathbb{E}_{\bigotimes_{i=1}^n \ket{\psi_i}\sim {\rm Haar}_1^{\otimes n}} \left[ \left\lvert\left(\bigotimes_{i=1}^n \bra{\psi_i}\right) \Tilde{U}^\dagger W \Tilde{U} \left(\bigotimes_{i=1}^n \ket{\psi_i}\right) \right\rvert^2 \right]\right)\, ,
    \end{align}
    where the last equality uses that single-qubit stabilizer states form a $2$-design (compare, e.g.,~\cite{gross2007evenly}, or see~\cite{kueng2015qubit, webb2016clifford, zhu2017multiqubit} for a stronger statement).
\end{proof}

To facilitate the comparison between the HST cost and a locally scrambled risk, we next show how to rewrite a general locally scrambled risk:

\begin{lem}\label{lem:locally-scrambled-cost-explicit-expression}
    Let $n\in\mathbb{N}$.
    Let $\mathcal{P}$ be a locally scrambled ensemble of $n$-qubit quantum states, with $\mathcal{U}_{{\rm test}}$ the corresponding locally scrambled unitary ensemble.
    Then, for any $n$-qubit unitaries $U$ and $V$, using the shorthand $W = U^\dagger V$,
    \begin{equation}
        R_{\mathcal{P}} \left( U, V \right)
        =  1 - \frac{1}{6^n} \sum_{A \subseteq \{1, \ldots, n\}} \mathbb{E}_{\Tilde{U}\sim \mathcal{U}_{{\rm test}}} \left[ \norm{\Tr_{A^c}\left[ \Tilde{U}^\dagger W \Tilde{U}\right]}_F^2 \right]\, ,
    \end{equation}
    where $\Tr_{A^c}$ denotes partial trace over all systems with index not in the set $A$ and $\norm{\cdot}_F$ denotes the Frobenius norm (which is the norm induced by the Hilbert-Schmidt inner product).  
\end{lem}
\begin{proof}
    Throughout the proof, we use the shorthand $W = U^\dagger V$.
    We begin by noticing that, since $\mathcal{U}$ is locally scrambled, for any fixed $U_1,\ldots,U_n\in\mathcal{U}(\mathbb{C}^2)$, we have:
    \begin{align}
        R_{\mathcal{P}} \left( U, V \right)
        &= 1 - \mathbb{E}_{\ket{\Psi}\sim \mathcal{P}} \left[ \left\lvert \bra{\Psi} W \ket{\Psi} \right\rvert^2 \right]\\
        &= 1 - \mathbb{E}_{\Tilde{U}\sim\mathcal{U}_{{\rm test}}} \left[ \left\lvert \bra{0}^{\otimes n} \Tilde{U}^\dagger W \Tilde{U}\ket{0}^{\otimes n} \right\rvert^2 \right]\\
        &= 1 - \mathbb{E}_{\Tilde{U}\sim\mathcal{U}_{{\rm test}}} \left[ \left\lvert \bra{0}^{\otimes n} \left( \bigotimes_{i=1}^n U_i^\dagger \right) \Tilde{U}^\dagger W \Tilde{U} \left( \bigotimes_{i=1}^n U_i \right) \ket{0}^{\otimes n} \right\rvert^2 \right] \, . \label{eq:locally-scrambled-cost-fixed-tensor-product-inserted}
    \end{align}
    If we now take an expectation over tensor products of Haar-random single-qubit unitaries $\bigotimes_{i=1}^n U_i\sim {\rm Haar}_1^{\otimes n}$, this yields:
    \begin{align}
        R_{\mathcal{P}} \left( U, V \right)
        &= 1 - \mathbb{E}_{\bigotimes_{i=1}^n U_i\sim {\rm Haar}_1^{\otimes n}} \mathbb{E}_{\Tilde{U}\sim\mathcal{U}_{{\rm test}}} \left[ \left\lvert \bra{0}^{\otimes n} \left( \bigotimes_{i=1}^n U_i^\dagger \right) \Tilde{U}^\dagger W \Tilde{U} \left( \bigotimes_{i=1}^n U_i \right) \ket{0}^{\otimes n} \right\rvert^2 \right]\\
        &= 1 - \mathbb{E}_{\Tilde{U}\sim\mathcal{U}_{{\rm test}}} \mathbb{E}_{\bigotimes_{i=1}^n U_i\sim {\rm Haar}_1^{\otimes n}} \left[ \left\lvert \bra{0}^{\otimes n} \left( \bigotimes_{i=1}^n U_i^\dagger \right) \Tilde{U}^\dagger W \Tilde{U} \left( \bigotimes_{i=1}^n U_i \right) \ket{0}^{\otimes n} \right\rvert^2 \right]\\
        &= \mathbb{E}_{\Tilde{U}\sim\mathcal{U}_{{\rm test}}} \left[ 1 - \mathbb{E}_{\bigotimes_{i=1}^n \ket{\psi_i}\sim {\rm Haar}_1^{\otimes n}} \left[ \left\lvert \left(\bigotimes_{i=1}^n \bra{\psi_i}\right) \Tilde{U}^\dagger W \Tilde{U} \left(\bigotimes_{i=1}^n \ket{\psi_i}\right) \right\rvert^2 \right] \right] \label{eq:locally-scrambled-cost-random-tensor-product-inserted-overlap-notation}\\
        &= \mathbb{E}_{\Tilde{U}\sim\mathcal{U}_{{\rm test}}} \left[ 1 - \mathbb{E}_{\bigotimes_{i=1}^n \ket{\psi_i}\sim {\rm Haar}_1^{\otimes n}} \left[ \Tr \left[ \left( \bigotimes_{i=1}^n \ket{\psi_i} \bra{\psi_i } \right)  \Tilde{U}^\dagger W \Tilde{U} \left( \bigotimes_{i=1}^n \ket{\psi_i} \bra{\psi_i } \right) \Tilde{U}^\dagger W^\dagger \Tilde{U} \right] \right] \right] \, ,\label{eq:locally-scrambled-cost-random-tensor-product-inserted-trace-notation}
    \end{align}
    where used that we can exchange the order of the expectation values by Tonelli's theorem, since the integrand is non-negative, and then slightly abused notation by using ${\rm Haar}_1^{\otimes n}$ to also denote the probability distribution describing a tensor product of Haar-random single-qubit states.
    
    Next, we recall the following Haar identity (see, e.g., Eq.~(2.26) in~\cite{roberts2017chaos}):
    \begin{align}
        \mathbb{E}_{\ket{\psi}\sim{\rm Haar}_1}\left[
        \ket{\psi}\bra{\psi}^{\otimes 2} \right] 
        &= \frac{\mathds{1} \otimes \mathds{1} + \operatorname{SWAP}}{6} \, ,\label{eq:basic-haar-integral}
    \end{align}
    where $\operatorname{SWAP}$ is the swap operator between two qubits and $\mathds{1}$ denotes the identity matrix on a single qubit system.
    Using this identity and the trace equality
    \begin{align}
        \Tr \left[ \operatorname{SWAP} (A\otimes B) \right]
        &= \Tr \left[ AB\right] \, ,
    \end{align}
    we can rewrite, for any fixed $\Tilde{U}\in\mathcal{U}((\mathbb{C}^{2})^{\otimes n})$,
    \begin{align} 
        & \mathbb{E}_{\bigotimes_{i=1}^n \ket{\psi_i} \sim {\rm Haar}_1^{\otimes n}}  \left[ \Tr \left[ \left( \bigotimes_{i=1}^n \ket{\psi_i} \bra{\psi_i } \right)  \Tilde{U}^\dagger W \Tilde{U} \left( \bigotimes_{i=1}^n \ket{\psi_i} \bra{\psi_i } \right) \Tilde{U}^\dagger W^\dagger \Tilde{U}  \right] \right]\\
        &= \mathbb{E}_{\bigotimes_{i=1}^n \ket{\psi_i} \sim {\rm Haar}_1^{\otimes n}}  \left[ \Tr \left[ \left(\bigotimes_{i=1}^n \operatorname{SWAP}_{i,i}\right) \left(\left( \bigotimes_{i=1}^n \ket{\psi_i} \bra{\psi_i } \right) \Tilde{U}^\dagger W \Tilde{U} \right)\otimes \left(\left( \bigotimes_{i=1}^n \ket{\psi_i} \bra{\psi_i } \right) \Tilde{U}^\dagger W^\dagger \Tilde{U}\right) \right] \right]\\
        &= \Tr\left[ \bigotimes_{i=1}^n \left( \operatorname{SWAP}_{i,i} \frac{\mathds{1}_i \otimes \mathds{1}_{i} + \operatorname{SWAP}_{i,i}}{6} \right) (\Tilde{U}^\dagger W \Tilde{U} \otimes \Tilde{U}^\dagger W^\dagger \Tilde{U}) \right]\, .
    \end{align}
    Here, $\operatorname{SWAP}_{i,i}$ denotes the single-qubit swap operator that acts on the $i^{\textrm{th}}$ qubits of the first and second $n$-qubit tensor factors, respectively.
    Next, we use that the swap operator is its own inverse and the definition of the partial trace to continue the computation:
    \begin{align}
        &\Tr\left[ \bigotimes_{i=1}^n \left( \operatorname{SWAP}_{i,i} \frac{\mathds{1}_i \otimes \mathds{1}_i + \operatorname{SWAP}_{i,i}}{6} \right) (\Tilde{U}^\dagger W \Tilde{U} \otimes \Tilde{U}^\dagger W^\dagger \Tilde{U}) \right]\\
        &= \Tr\left[ \bigotimes_{i=1}^n \left( \frac{\mathds{1}_i \otimes \mathds{1}_i + \operatorname{SWAP}_{i,i}}{6} \right) (\Tilde{U}^\dagger W \Tilde{U} \otimes \Tilde{U}^\dagger W^\dagger \Tilde{U}) \right] \\
        &= \frac{1}{6^n} \sum_{A \subseteq \{1, \ldots, n\}} \Tr\left[ \left(\left(\bigotimes_{i\not\in A}(\mathds{1}_i \otimes \mathds{1}_i)\right)\otimes\left(\bigotimes_{i\in A}\operatorname{SWAP}_{i,i}\right)\right)(\Tilde{U}^\dagger W \Tilde{U} \otimes \Tilde{U}^\dagger W^\dagger \Tilde{U})\right]\\
        &= \frac{1}{6^n} \sum_{A \subseteq \{1, \ldots, n\}} \Tr\left[ \left(\bigotimes_{i\in A}\operatorname{SWAP}_{i,i}\right) \Tr_{A^c,A^c}\left[\Tilde{U}^\dagger W \Tilde{U} \otimes \Tilde{U}^\dagger W^\dagger \Tilde{U}\right]\right]\, .
    \end{align}
    In the last step, we have used $\Tr_{A^c,A^c}$ to denote the partial trace over the $A^c$-subsystems of both the first and the second $n$-qubit tensor factors. 
    Observing that $\Tr_{A^c,A^c}\left[\Tilde{U}^\dagger W \Tilde{U} \otimes \Tilde{U}^\dagger W^\dagger \Tilde{U}\right] = \Tr_{A^c}[\Tilde{U}^\dagger W \Tilde{U}]\otimes\Tr_{A^c}[\Tilde{U}^\dagger W^\dagger \Tilde{U}]$ and writing $\operatorname{SWAP}_{A,A} = \bigotimes_{i\in A}\operatorname{SWAP}_{i,i}$, we finish the computation and obtain
    \begin{align}
        &\frac{1}{6^n} \sum_{A \subseteq \{1, \ldots, n\}} \Tr\left[ \left(\bigotimes_{i\in A}\operatorname{SWAP}_{i,i}\right) \Tr_{A^c,A^c}\left[\Tilde{U}^\dagger W \Tilde{U} \otimes \Tilde{U}^\dagger W^\dagger \Tilde{U}\right]\right]\\
        &= \frac{1}{6^n} \sum_{A \subseteq \{1, \ldots, n\}} \Tr\left[ \operatorname{SWAP}_{A,A}(\Tr_{A^c}[\Tilde{U}^\dagger W \Tilde{U}]\otimes\Tr_{A^c}[\Tilde{U}^\dagger W^\dagger \Tilde{U}])\right]\\
        &= \frac{1}{6^n} \sum_{A \subseteq \{1, \ldots, n\}} \Tr\left[ \Tr_{A^c}\left[\Tilde{U}^\dagger W \Tilde{U}\right] \Tr_{A^c}\left[\Tilde{U}^\dagger W^\dagger \Tilde{U} \right] \right] \\
        &= \frac{1}{6^n} \sum_{A \subseteq \{1, \ldots, n\}} \norm{\Tr_{A^c}\left[\Tilde{U}^\dagger W \Tilde{U}\right]}_F^2 \, .
    \end{align}
    Plugging this observation into Eq.~\eqref{eq:locally-scrambled-cost-random-tensor-product-inserted-trace-notation} finishes the proof.
\end{proof}

As an aside, note that the proof of Lemma~\ref{lem:locally-scrambled-cost-explicit-expression} does not use the unitarity of $U$ and $V$. That is, Lemma~\ref{lem:locally-scrambled-cost-explicit-expression}, just like Lemma~\ref{lem:hst-cost-pauli-average}, is valid for arbitrary $U,V\in\mathcal{B}(\mathbb{C}^d)$. However, Lemma~\ref{lem:auxiliary-overlap-bounds}, Corollary~\ref{corol:hst-cost-bound-local-haar-average}, and the results from here on make use of the unitarity assumption.


Now, we are ready to combine the tools developed so far to establish the main technical result of this subsection.
We show that all locally scrambled risks are equivalent to the HST cost:

\begin{lem}[Restatement of Lemma~\ref{lem:locally-scrambled-cost-versus-hst-cost}]\label{lem:supplementary-locally-scrambled-cost-versus-hst-cost}
    Let $n\in\mathbb{N}$ and write $d=2^n$.
    Let $\mathcal{P}$ be a locally scrambled ensemble of $n$-qubit quantum states, with $\mathcal{U}_{{\rm test}}$ the corresponding locally scrambled unitary ensemble.
    Then, for any $n$-qubit unitaries $U\in\mathcal{U}((\mathbb{C}^{2})^{\otimes n})$ and $V\in\mathcal{U}((\mathbb{C}^{2})^{\otimes n})$,
    \begin{equation}
        \frac{1}{2} C_{\mbox {\tiny \rm HST}}(U, V)
        \leq R_{\mathcal{P}} \left( U, V \right)
        \leq C_{\mbox {\tiny \rm HST}}(U, V) \, .
    \end{equation}
\end{lem}
\begin{proof}
    Throughout the proof, we use the shorthand $W = U^\dagger V$.
    We start with the proof of the first inequality, $\frac{1}{2} C_{\mbox {\tiny \rm HST}}(U, V)\leq R_{\mathcal{P}} \left( U, V \right)$.
    To this end, we apply Corollary~\ref{corol:hst-cost-bound-local-haar-average} to see that, for any $\tilde{U}\in\mathcal{U}((\mathbb{C}^{2})^{\otimes n})$,
    \begin{equation}
        C_{\mbox {\tiny \rm HST}}(U, V)
        \leq  2\left( 1 - \mathbb{E}_{\bigotimes_{i=1}^n \ket{\psi_i}\sim {\rm Haar}_1^{\otimes n}} \left[ \left\lvert\left(\bigotimes_{i=1}^n \bra{\psi_i}\right) \Tilde{U}^\dagger W \Tilde{U} \left(\bigotimes_{i=1}^n \ket{\psi_i}\right) \right\rvert^2 \right]\right)\, .
    \end{equation}
    Taking an expectation over $\Tilde{U}\sim \mathcal{U}_{{\rm test}}$, we obtain:
    \begin{align}
        C_{\mbox {\tiny \rm HST}}(U, V)
        &=  \mathbb{E}_{\Tilde{U}\sim\mathcal{U}_{{\rm test}}} \left[C_{\mbox {\tiny \rm HST}}(\Tilde{U}^\dagger U \Tilde{U}, \Tilde{U}^\dagger V \Tilde{U})\right] \\
        &\leq 2  \mathbb{E}_{\Tilde{U}\sim\mathcal{U}_{{\rm test}}} \left[ 1 - \mathbb{E}_{\bigotimes_{i=1}^n \ket{\psi_i}\sim {\rm Haar}_1^{\otimes n}} \left[ \left\lvert \left(\bigotimes_{i=1}^n \bra{\psi_i}\right) \Tilde{U}^\dagger W \Tilde{U} \left(\bigotimes_{i=1}^n \ket{\psi_i}\right) \right\rvert^2 \right] \right]
        \\&= 2 R_{\mathcal{P}} \left( U, V \right)\, ,
    \end{align}
    where the last equality uses that $\mathcal{U}$ is locally scrambled and was already derived previously in Eq.~\eqref{eq:locally-scrambled-cost-random-tensor-product-inserted-overlap-notation}.
    This finishes the proof of the first inequality.
    
    \medskip
    
    We now turn our attention to the second inequality, $R_{\mathcal{P}} \left( U, V \right)\leq C_{\mbox {\tiny \rm HST}}(U, V)$. 
    To prove this inequality, we rely on the expression for $R_{\mathcal{P}} \left( U, V \right)$ derived in Lemma~\ref{lem:locally-scrambled-cost-explicit-expression}.
    So, let $A\subseteq\{1,\ldots,n\}$ be a subset of cardinality $\lvert A\rvert = k$. 
    Then, we have, again for any fixed $\Tilde{U}\in\mathcal{U}(\mathbb{C}^{d})$:
    \begin{align}
        \norm{\Tr_{A^c}\left[ \Tilde{U}^\dagger W \Tilde{U}\right]}_F^2
        &= \sum_{i=1}^{2^k} \left( s_i \left(\Tr_{A^c}\left[\Tilde{U}^\dagger W \Tilde{U}\right]\right) \right)^2\\
        &\geq \frac{1}{2^k} \left( \sum_{i=1}^{2^k} s_i \left(\Tr_{A^c}\left[\Tilde{U}^\dagger W \Tilde{U}\right]\right) \right)^2\\
        &= \frac{1}{2^k} \norm{\Tr_{A^c}\left[\Tilde{U}^\dagger W \Tilde{U}\right]}_1^2\\
        &\geq \frac{1}{2^k} \left\lvert \Tr \left[ \Tr_{A^c}\left[\Tilde{U}^\dagger W \Tilde{U}\right] \right]\right\rvert^2\\
        &= \frac{1}{2^k} \left\lvert \Tr \left[ \Tilde{U}^\dagger W \Tilde{U} \right]\right\rvert^2\\
        &= \frac{1}{2^k} \left\lvert \Tr \left[  W \right]\right\rvert^2\, ,
    \end{align}
    where the first line is the definition of the Frobenius norm (also known as Schatten $2$-norm) as $2$-norm of the vector of singular values, the second line uses Jensen's inequality, the third line uses the definition of the trace norm (also known as Schatten $1$-norm) as the sum of singular values, the fourth line uses Hölder's inequality (for $p=1$ and $q=\infty$, i.e. $\Tr[AB] \leq ||A||_1 ||B||_\infty$ with $B = \I$), the fifth line uses that the trace of a partial trace equals the trace of the original matrix, and the last line uses unitarity of $\Tilde{U}$ and the basis-invariance of the trace.
    Taking an expectation over $\Tilde{U}\sim \mathcal{U}_{{\rm test}}$ and combining the resulting inequality with Lemma~\ref{lem:locally-scrambled-cost-explicit-expression}, we get
    \begin{align}
        R_{\mathcal{P}} \left( U, V \right)
        &= 1 - \frac{1}{6^n} \sum_{A \subseteq \{1, \ldots, n\}} \mathbb{E}_{\Tilde{U}\sim \mathcal{U}_{{\rm test}}} \left[ \norm{\Tr_{A^c}\left[ \Tilde{U}^\dagger W \Tilde{U}\right]}_F^2 \right]\\
        &\leq 1 - \frac{1}{6^n} \sum_{k=0}^n \binom{n}{k}\cdot \frac{1}{2^{k}}\lvert\Tr\left[W\right]\rvert^2 \\
        &=  1 - \frac{1}{4^n} \lvert\Tr\left[W\right] \rvert^2\\
        &= C_{\mbox {\tiny \rm HST}}(U, V) \, .
    \end{align}
    This is the second inequality that we set out to prove.
\end{proof}

As an immediate consequence of Lemma~\ref{lem:supplementary-locally-scrambled-cost-versus-hst-cost}, since all locally scrambled risks are equivalent to the HST cost, we also see that all locally scrambled risks are equivalent to each other up to a constant multiplicative factor. That is, the following holds.

\begin{thm}[Equivalence of locally scrambled ensembles for comparing unitaries -- Restatement of Theorem~\ref{thm:comparing-different-locally-scrambled-costs}]\label{thm:supplementary-comparing-different-locally-scrambled-costs}
    Let $\mathcal{P}$ and $\mathcal{Q}$ be two locally scrambled ensembles of $n$-qubit quantum states.
    Then, for any $n$-qubit unitaries $U$ and $V$,
    \begin{equation}
        \frac{1}{2} R_{\mathcal{Q}} \left( U, V \right)
        \leq R_{\mathcal{P}} \left( U, V \right)
        \leq 2 R_{\mathcal{Q}} \left( U, V \right) \, .
    \end{equation}
\end{thm}

Thus, for the purposes of comparing unitaries, all locally scrambled ensembles are in effect equivalent.
In the next subsection, we combine this insight with known in-distribution generalization bounds for learning unitaries via QNNs to establish out-of-distribution generalization guarantees, if both the training and the testing distribution are locally scrambled.
In particular, we will use Theorem~\ref{thm:supplementary-comparing-different-locally-scrambled-costs} to show that, if we train on input states coming from a ``simple'' locally scrambled ensemble, such as random product states (Example~\ref{ex:product-haar-random}), and have good in-distribution-generalization there, then we generalize to any other ``more complicated'' locally scrambled ensemble, such as fully Haar-random (and thus highly entangled) states or output states of random circuits (Examples~\ref{ex:global-haar-random} and~\ref{ex:random-circuit-outputs}).

\begin{rmk}\label{remark:second-moments-suffice}
    By definition, the expected risk in Eq.~\eqref{eq:supplementary-test-cost-general} depends only on the complex second moment of the testing distribution $\mathcal{P}$.
    Therefore, we can directly extend Theorem~\ref{thm:supplementary-comparing-different-locally-scrambled-costs} beyond locally scrambled ensembles, i.e. elements of $\mathbb{S}_{\rm LS}$, to ensembles whose complex second moments agree with those of some locally scrambled ensembles, i.e. to elements of $\mathbb{S}_{\rm LS}^{(2)}$.
    As a concrete example (discussed also in Example~\ref{ex:product-random-stabilizer}): Single-qubit stabilizer states form a $2$-design and tensor products of Haar-random single-qubit states are locally scrambled. Thus, also the testing risk obtained by taking tensor products of random single-qubit stabilizer states for $\mathcal{P}$ in Eq.~\eqref{eq:supplementary-test-cost-general} is equivalent to any locally scrambled risk up to factors of $2$.
\end{rmk}

\subsection{Out-of-Distribution Generalization for QNNs Trained on Locally Scrambled States}\label{ap:genproofs}

In this subsection, we use our results of Subsection~\ref{subsection:supplementary-equivalence-locally-scrambled-costs} to strengthen in-distribution generalization bounds for learning unitaries via QNNs to out-of-distribution generalization bounds for the same task, where we allow for arbitrary locally scrambled training and testing distributions. 
The following theorem, where we denote the in-distribution generalization error by $\operatorname{gen}_{\mathcal{Q}, \mathcal{D}_{\mathcal{Q}}(N)} \left(\alv \right) = R_\mathcal{Q} \left(U, V(\alv) \right) - C_{\mathcal{D}_{\mathcal{Q}}(N)} \left(U, V(\alv) \right)$, serves as a general template for such a strengthening:

\begin{corol}[Out-of-distribution generalization from in-distribution generalization in unitary learning -- Restatement of Corollary~\ref{corol:ood-generalization-from-id-generalization-unitary-learning}]\label{corol:supplementary-ood-generalization-from-id-generalization-unitary-learning}
    Let $n\in\mathbb{N}$ and write $d=2^n$.
    Let $\mathcal{Q}$ and $\mathcal{P}$ be two locally scrambled ensembles of $n$-qubit quantum states.
    Let $U$ be an unknown $n$-qubit unitary.
    Let $V(\alv)$ be an $n$-qubit unitary QNN.
    For any parameter setting $\alv$, we have
    \begin{equation}
        R_{\mathcal{P}} (U, V(\alv))
        \leq 2\left( C_{\mathcal{D}_{\mathcal{Q}}(N)} (U, V(\alv)) + \operatorname{gen}_{\mathcal{Q}, \mathcal{D}_{\mathcal{Q}}(N)} \left(\alv \right) \right) \, .
    \end{equation}
\end{corol}
\begin{proof}
    As both $\mathcal{Q}$ and $\mathcal{P}$ are locally scrambled ensembles of $n$-qubit quantum states, Theorem~\ref{thm:supplementary-comparing-different-locally-scrambled-costs} yields
    \begin{equation}
        R_{\mathcal{P}} (U, V(\alv))
        \leq 2 R_{\mathcal{Q}} (U, V(\alv))\, .
    \end{equation}
    After rewriting
    \begin{equation}
        R_{\mathcal{Q}} (U, V(\alv))
        = C_{\mathcal{D}_{\mathcal{Q}}(N)} (U, V(\alv)) + \left( R_\mathcal{Q} \left( U, V(\alv) \right) - C_{\mathcal{D}_{\mathcal{Q}}(N)} \left( U, V(\alv) \right) \right) \, ,
    \end{equation}
    this gives the statement of the theorem.
\end{proof}

Corollary~\ref{corol:supplementary-ood-generalization-from-id-generalization-unitary-learning} has the following implication for out-of-distribution generalization after training: 
When training the QNN $V(\alv)$ with the cost $C_{\mathcal{D}_{\mathcal{Q}}(N)}$ using training data $\mathcal{D}_{\mathcal{Q}}(N)$, the out-of-distribution testing risk $R_{\mathcal{P}} (U, V(\alv_{\rm opt}))$ of the parameter setting $\alv_{\rm opt}$ after training is controlled in terms of the training cost and the in-distribution generalization error.
Here, the only assumption on the training and testing distributions is that they are both locally scrambled.

To demonstrate the usefulness of Corollary~\ref{corol:supplementary-ood-generalization-from-id-generalization-unitary-learning}, we next show the concrete form it takes when combined with the QNN generalization guarantees of~\cite{caro2021generalization}:

\begin{corol}[Locally scrambled out-of-distribution generalization for QNNs - Restatement of Corollary~\ref{corol:locally-scrambled-ood-generalization-qnn}]\label{corol:supplementary-locally-scrambled-ood-generalization-qnn}
    Let $n\in\mathbb{N}$ and write $d=2^n$. Let $\delta\in (0,1)$.
    Let $\mathcal{Q}$ and $\mathcal{P}$ be two locally scrambled ensembles of $n$-qubit quantum states.
    Let $U$ be an unknown $n$-qubit unitary.
    Let $V(\alv)$ be an $n$-qubit unitary QNN with $T$ parameterized local gates.
    When trained with the cost $C_{\mathcal{D}_{\mathcal{Q}}(N)}$ using training data $\mathcal{D}_{\mathcal{Q}}(N)$, the out-of-distribution testing risk w.r.t.~$\mathcal{P}$ of the parameter setting $\alv_{\rm opt}$ after training satisfies
    \begin{equation}
        R_{\mathcal{P}} (U, V(\alv_{\rm opt}))
        \leq 2 C_{\mathcal{D}_{\mathcal{Q}}(N)} (U, V(\alv_{\rm opt})) + \mathcal{O} \left( \sqrt{\frac{T \log (T)}{N}} + \sqrt{\frac{\log (\nicefrac{1}{\delta})}{N}}\right)\, ,
    \end{equation}
    with probability $\geq 1-\delta$ over the choice of training data of size $N$ according to $\mathcal{Q}$.
\end{corol}
\begin{proof}
    According to~\cite[Theorem 11]{caro2021generalization}, the in-distribution generalization error in this setting is bounded as
    \begin{equation}
        R_\mathcal{Q} \left( U, V(\alv) \right) - C_{\mathcal{D}_{\mathcal{Q}}(N)} \left( U, V(\alv) \right)
        \leq \mathcal{O} \left( \sqrt{\frac{T \log (T)}{N}} + \sqrt{\frac{\log (\nicefrac{1}{\delta})}{N}}\right) \, .
    \end{equation}
    Plugging this in-distribution generalization error bound into Corollary~\ref{corol:supplementary-ood-generalization-from-id-generalization-unitary-learning} yields the claim.
\end{proof}

Corollary~\ref{corol:supplementary-locally-scrambled-ood-generalization-qnn} has the following implication for training data requirements:
To ensure that, with high probability, the expected risk does not exceed twice the training cost by more than a specified accuracy $\varepsilon$, it suffices to have training data of size $\sim \nicefrac{T\log (T)}{\varepsilon^2}$. This sufficient training data size scales only slightly superlinearly in the number of trainable gates in the QNN.

For the special case of the locally scrambled training ensemble being tensor products of Haar-random states, with corresponding product state training data given as
\begin{equation}\label{eq:supplementary-training-data-product}
    \mathcal{D}_{{\rm Haar}_1^{\otimes n}}(N)
    = \{(\ket{\Psi_{{\rm Haar}_1^{\otimes n}}^{(j)}},\ket{\Phi^{(j)}_{{\rm Haar}_1^{\otimes n}}}) \}_{j=1}^{N}
    = \left\{\left(\bigotimes_{i=1}^n \ket{\psi^{(j)}_i}, U\left(\bigotimes_{i=1}^n \ket{\psi^{(j)}_i}\right)\right)\right\}_{j=1}^N \, ,
\end{equation}
where the $\ket{\psi^{(j)}_i}$ are independent Haar-random single-qubit states and $U$ is the unknown target unitary, Corollary~\ref{corol:supplementary-locally-scrambled-ood-generalization-qnn} becomes:

\begin{corol}[Out-of-distribution generalization for QNNs trained on random product states]\label{corol:supplementary-ood-generalization-random-product-states-qnn}
    Let $n\in\mathbb{N}$ and write $d=2^n$. Let $\delta\in (0,1)$.
    Let $\mathcal{P}$ be a locally scrambled ensemble of $n$-qubit quantum states.
    Let $U$ be an unknown $n$-qubit unitary.
    Let $V(\alv)$ be an $n$-qubit unitary QNN with $T$ parameterized local gates.
    When trained with the cost $C_{\mathcal{D}_{{\rm Haar}_1^{\otimes n}}(N)}$ using training data $\mathcal{D}_{{\rm Haar}_1^{\otimes n}}(N)$, the out-of-distribution risk w.r.t.~$\mathcal{P}$ of the parameter setting $\alv_{\rm opt}$ after training satisfies
    \begin{equation}
        R_{\mathcal{P}} (U, V(\alv_{\rm opt}))
        \leq 2 C_{\mathcal{D}_{{\rm Haar}_1^{\otimes n}}(N)} (U, V(\alv_{\rm opt})) + \mathcal{O} \left( \sqrt{\frac{T \log (T)}{N}} + \sqrt{\frac{\log (\nicefrac{1}{\delta})}{N}}\right)\, ,
    \end{equation}
    with probability $\geq 1-\delta$ over the choice of training data of size $N$ according to $\mathcal{S}_{{\rm Haar}_1^{\otimes n}}$.
\end{corol}

Notice that, following Remark~\ref{remark:second-moments-suffice}, we can simplify the training data even further and still achieve the same performance.
Namely, if we train on tensor products of random stabilizer states (instead of on tensor products of random product states), then exactly the same out-of-distribution generalization bound as in Corollary~\ref{corol:supplementary-ood-generalization-random-product-states-qnn} holds.

\begin{rmk}\label{remark:local}
    We can extend our results for out-of-distribution generalization when training on tensor products of Haar-random states to local variants of our risks and costs. 
    Such local costs are essential to avoid cost function dependent barren plateaus~\cite{cerezo2021cost} when training a shallow QNN, thereby facilitating optimization.
    As a concrete example, when taking $\mathcal{S}_{{\rm Haar}_1^{\otimes n}}$ from Example~\ref{ex:product-haar-random} as testing ensemble, we can consider the local expected testing risk
    \begin{equation}
        R_{\mathcal{S}_{{\rm Haar}_1^{\otimes n}}}^{L}(U, V(\alv)) 
        = 1 - \mathbb{E}_{\ket{\Psi_{\rm P}}=\bigotimes_{i=1}^n \ket{\psi_i} \sim {\mathcal{S}_{{\rm Haar}_1^{\otimes n}}}}\left[ \frac{1}{n}\sum\limits_{i=1}^n  \Tr \left[ U \ket{\Psi_{\rm P}^{(j)}}\bra{\Psi_{\rm P}^{(j)}}  U^\dagger V(\alv) \left( \ket{\psi_i^{(j)}} \bra{\psi_i^{(j)} }\otimes \mathds{1}_{\bar{i}} \right)  V(\alv)^\dagger\right]\right]\label{eq:supplementary-product-local-cost}
    \end{equation}
    and, for a training data set as in Eq.~\eqref{eq:supplementary-training-data-product}, the local training cost
    \begin{equation}\label{eq:supplementary-product-local-training-cost}
        C_{ \mathcal{D}_{{\rm Haar}_1^{\otimes n}}(N)}^{L}(U, V(\alv))
        = 1 -\frac{1}{nN}\sum\limits_{j=1}^N\sum\limits_{i=1}^n \Tr \left[ U \ket{\Psi_P^{(j)}}\bra{\Psi_P^{(j)}}  U^\dagger V(\alv)\left(\ket{\psi_i^{(j)}} \bra{\psi_i^{(j)} }\otimes \mathds{1}_{\bar{i}}\right) V(\alv)^\dagger\right] \, .
    \end{equation}
    Clearly, analogous local variants of expected risk and training cost can be defined whenever the respective ensemble has a tensor product structure. Among the examples presented in the main text, both $\mathcal{S}_{{\rm Haar}_1^{\otimes n}}$ from Example~\ref{ex:product-haar-random} and $\mathcal{S}_{{\rm Stab}_1^{\otimes n}}$ from Example~\ref{ex:product-random-stabilizer} have that form. However, if the training data is highly entangled constructing such local costs in this manner is not possible. Thus, this is another important consequence of our proof that training on product state training data enjoys out-of-distribution generalization.
    
    According to~\cite[Appendix C]{khatri2019quantum}, we know that
    \begin{equation}
        R_{\mathcal{S}_{{\rm Haar}_1^{\otimes n}}}^{L}(U, V(\alv))
        \leq R_{\mathcal{S}_{{\rm Haar}_1^{\otimes n}}} (U, V(\alv))
        \leq n\cdot R_{\mathcal{S}_{{\rm Haar}_1^{\otimes n}}}^{L}(U, V(\alv)) \, .
    \end{equation}
    We can combine this with Theorem~\ref{thm:supplementary-comparing-different-locally-scrambled-costs} to obtain: 
    If $\mathcal{P}$ is any locally scrambled ensemble of $n$-qubit quantum states, then for any $n$-qubit unitaries $U$ and $V$,
    \begin{equation}
        \frac{1}{2} R_{\mathcal{S}_{{\rm Haar}_1^{\otimes n}}}^{L}(U, V(\alv))
        \leq R_{\mathcal{P}} \left( U, V(\alv) \right)
        \leq 2 n R_{\mathcal{S}_{{\rm Haar}_1^{\otimes n}}}^{L}(U, V(\alv)) \, .
    \end{equation}
    With this observation, we can obtain Corollary~\ref{corol:locally-scrambled-ood-generalization-qnn-local-cost} as a version of Corollary~\ref{corol:supplementary-ood-generalization-random-product-states-qnn} when training on products of Haar-random single-qubit states, but now with the local cost from Eq.~\eqref{eq:supplementary-product-local-training-cost}, simply replacing $C_{\mathcal{D}_{{\rm Haar}_1^{\otimes n}}(N)} (U, V(\alv_{\rm opt}))$ by $n C_{\mathcal{D}_{{\rm Haar}_1^{\otimes n}}(N)}^{L} (U, V(\alv_{\rm opt}))$ and $\mathcal{O} \left( \sqrt{\frac{T \log (T)}{N}} + \sqrt{\frac{\log (\nicefrac{1}{\delta})}{N}}\right)$ by $\mathcal{O} \left( n\sqrt{\frac{T \log (T)}{N}} + \sqrt{\frac{\log (\nicefrac{1}{\delta})}{N}}\right)$.
\end{rmk}

\subsection{Remarks on the Role of Linearity}\label{ap:role-of-linearity}

As outlined in the discussion between Corollary \ref{corol:locally-scrambled-ood-generalization-qnn} and Corollary \ref{corol:locally-scrambled-ood-generalization-qnn-local-cost}, linearity is important in enabling our out-of-distribution generalization.  
Intuitively, as long as the training states span the space on which one wishes to learn the action of the target unitary, it ought to be possible to train on those states and by linearity extrapolate to the entire space. 
However, this line of argument alone is insufficient to explain out-of-distribution generalization. The random ensembles of states also have to be ``well-behaved'' to ensure good generalization from a manageable number of training states.

One way of highlighting this subtlety is to note that even an exponential number of computational basis states cannot be used to learn an unknown unitary using a cost formulated in terms of the $1$-norm distance between the guess output and true output (or equivalently the fidelity between the guess and true outputs). Namely, computational basis states do not allow to learn relative phases.
This can be illustrated by the following concrete example: Suppose the unknown unitary is the single-qubit unitary $U=e^{-i\varphi Z}$ for some $\varphi\in [0,2\pi)$. That is, we consider 
\begin{equation}
    U=e^{-i\varphi}\begin{pmatrix}  1 & 0 \\ 0 & e^{2i\varphi}\end{pmatrix} \, .
\end{equation}
The action of $U$ on the computational basis states is thus given by $U\ket{0} = e^{-i\varphi}\ket{0}$ and $U\ket{1} = e^{i\varphi}\ket{1}$, all the relevant information lies in the relative phase between the two output states. 
As our notions of risk are (as is physically reasonable) independent of global phases in the output states, the unitary $V=\mathds{1}_2$ achieves a perfect training cost on the training data set $\mathcal{D}=\{(\ket{0},U\ket{0}), (\ket{1}, U\ket{1})\}$, namely
\begin{align}
    C_{\mathcal{D}} (U,V)
    &= \frac{1}{8}\left(\norm{U\ket{0}\bra{0}U^\dagger - V\ket{0}\bra{0}V^\dagger}_1^2 + \norm{U\ket{1}\bra{1}U^\dagger - V\ket{1}\bra{1}V^\dagger}_1^2\right)\\
    &= \frac{1}{8}\left(\norm{\ket{0}\bra{0} - \ket{0}\bra{0}}_1^2 + \norm{\ket{1}\bra{1} - \ket{1}\bra{1}}_1^2\right)\\
    &= 0\, ,
\end{align}
and therefore, since the training data set in this case consists of exactly the two single-qubit computational basis states, a perfect expected testing risk over randomly drawn computational basis states, namely
\begin{align}
    R_{\mathcal{S}_{\mathrm{CompBasis}}} (U,V) 
    &= \frac{1}{4}\mathbb{E}_{\ket{\Psi}\sim \mathcal{S}_{\mathrm{CompBasis}}}\left[\norm{U\ket{\Psi}\bra{\Psi}U^\dagger - V\ket{\Psi}\bra{\Psi}V^\dagger}_1^2\right]\\
    &= \frac{1}{8}\left(\norm{U\ket{0}\bra{0}U^\dagger - V\ket{0}\bra{0}V^\dagger}_1^2 + \norm{U\ket{1}\bra{1}U^\dagger - V\ket{1}\bra{1}V^\dagger}_1^2\right)\\
    &= C_{\mathcal{D}} (U,V)\\
    &= 0 \, .
\end{align}
However, $V$ clearly fails to capture the relative phase between $U\ket{0}$ and $U\ket{1}$. In particular, if we consider the testing risk over uniformly random states in the $X$-basis, we see that
\begin{align}
    R_{\mathcal{S}_{X-\mathrm{Basis}}} (U,V) 
    &= \frac{1}{4}\mathbb{E}_{\ket{\Psi}\sim \mathcal{S}_{X-\mathrm{Basis}}}\left[\norm{U\ket{\Psi}\bra{\Psi}U^\dagger - V\ket{\Psi}\bra{\Psi}V^\dagger}_1^2\right]\\
    &= \frac{1}{8}\left(\norm{U\ket{+}\bra{+}U^\dagger - V\ket{+}\bra{+}V^\dagger}_1^2 + \norm{U\ket{-}\bra{-}U^\dagger - V\ket{-}\bra{-}V^\dagger}_1^2\right)\\
    &= \frac{1}{8}\left(\norm{U\ket{+}\bra{+}U^\dagger - \ket{+}\bra{+}}_1^2 + \norm{U\ket{-}\bra{-}U^\dagger - \ket{-}\bra{-}}_1^2\right)\\
    &= \frac{1}{8}\left( \norm{\frac{1}{2}\left((e^{2i\varphi}-1)\ket{1}\bra{0} + (e^{-2i\varphi}-1)\ket{0}\bra{1}\right)}_1^2 + \norm{\frac{1}{2}\left((1 - e^{2i\varphi})\ket{1}\bra{0} + (1 - e^{-2i\varphi})\ket{0}\bra{1}\right)}_1^2 \right)\\
    &= \frac{1}{16} \norm{ \begin{pmatrix} 0 & e^{-2i\varphi}-1\\ e^{2i\varphi}-1 & 0 \end{pmatrix}}_1^2\\
    &= \frac{1-\cos(2\varphi)}{2}\\
    &= \sin^2(\varphi) \, , 
\end{align}
which is strictly bigger than zero whenever $\varphi$ is not an integer multiple of $\pi$.
Also, using Eq.~\eqref{eq:hst-cost-versus-expected-haar-cost}, we obtain
\begin{align}
    R_{\mathcal{S}_{\rm Haar}}(U,V)
    = \frac{2}{3} \left(1 - \frac{1}{4} \lvert \tr[U^\dagger V]\rvert^2\right)
    = \frac{2}{3}\sin^2(\varphi) \, ,
\end{align}
which is strictly bigger than zero whenever $\varphi$ is not an integer multiple of $\pi$.
Thus, despite perfect training error, perfect in-distribution generalization error, and perfect in-distribution testing error, the out-of-distribution generalization and testing errors can be non-zero.
This single-qubit example shows that no analogue of Theorem \ref{thm:comparing-different-locally-scrambled-costs} can hold without additional assumptions on shared properties between the training and testing ensembles (such as both being locally scrambled). In particular, a good training and testing performance on randomly drawn computational basis states does not imply a good testing performance over (for example) random $X$-basis states or Haar-random states.
While this counterexample to out-of-distribution generalization from training on computational basis states is specific to our (physically motivated) choice of cost function, the above argument still emphasizes that linearity alone does not trivially imply out-of-distribution generalization.

It is further worth stressing that an argument based on linearity places no guarantees on how many training states are required/sufficient for convergence. The argument that `as long as the training states span the space on which you wish to learn the action of the target unitary on, it ought to be possible to train on those states and by linearity extrapolate to the entire space' crucially only applies if you train on an exponentially large training ensemble. In general, how many states are required/sufficient to ensure good generalization will depend on the types of states in the training ensemble. In our work, we combine Theorem \ref{thm:comparing-different-locally-scrambled-costs} with the recent in-distribution generalization bounds of \cite{caro2021generalization} and thereby show that the worst-case training data requirements for random product states cannot be significantly worse than those for fully random states.

\section{Additional Numerical Results}

\subsection{Numerical Test of Lemma~\ref{lem:locally-scrambled-cost-versus-hst-cost}}\label{ap:TightnessNumerics}

We numerically probe the validity of Lemma~\ref{lem:locally-scrambled-cost-versus-hst-cost} for the NISQ friendly scenario $\mathcal{Q}= \mathcal{S}_{\text{Haar}_1^{\otimes n}}$, i.e. the set of Haar-random product states. In this case, the relation between the average Haar-random product state cost and the general n-qubit Haar-random state cost according to Lemma~\ref{lem:locally-scrambled-cost-versus-hst-cost} is
\begin{equation}
    \label{eq:nisq-cost-comp}
    R_{\mathcal{S}_{{\rm Haar}_1^{\otimes n}}}(U, V(\vec{\alpha})) 
    \leq \frac{d + 1}{d}R_{\mathcal{S}_{\text{Haar}_n}}(U, V(\vec{\alpha})) \leq 2 R_{\mathcal{S}_{{\rm Haar}_1^{\otimes n}}}(U, V(\vec{\alpha})),
\end{equation}
where $\vec{\alpha}$ contains all parameters that define the QNN. To make contact with NISQ applications, we will probe the validity of Ineq.~\eqref{eq:nisq-cost-comp} by sampling $W = V^{\dagger} U(\vec{\alpha})$ from random low-depth quantum circuits. Given $W$, we can evaluate $\frac{d + 1}{d}R_{\mathcal{S}_{\text{Haar}_n}}(U, V(\vec{\alpha}))$ from Eq.~\ref{eq:hst-cost-versus-expected-haar-cost} and $R_{\mathcal{S}_{{\rm Haar}_1^{\otimes n}}}(U, V(\vec{\alpha}))$ from Lemma~\ref{lem:locally-scrambled-cost-explicit-expression} with straightforward matrix operations. To ensure that we can quickly sample from a large range of cost values without the need for optimization, we implement an ansatz of the form
\begin{equation}
    \label{eq:cost-comp-ansatz}
    W(\vec{\alpha} = r \cdot \vec{\theta}) = \prod_{k=1}^l \left( \prod_{i=1}^T G_{ik}(r \cdot \vec{\vec{\theta_{ik}}}) \right).
\end{equation}

Here, each $G_{ik}$ is an arbitrary 2-qubit gate, and the inner product over $i$ represents a hardware-efficient tiling of these 2-qubit gates. That is, we apply $\lfloor n / 2\rfloor$ 2-qubit gates from even qubits to odd qubits (i.e. between $(0, 1), (2, 3), \text{etc}\ldots)$ in parallel and then $\lfloor (n-1) / 2\rfloor$ 2-qubit gates from odd to even (i.e. between $(1, 2), (3, 4), \text{etc}\ldots)$ in parallel. The outer product just means we are applying $l$ layers. So far, we have just described a familiar hardware-efficient tiling of arbitrary 2-qubit gates used in many variational quantum algorithms~\cite{cerezo2021cost, cirstoiu2020variational, gibbs2021long}, but there is one crucial difference between our implementation and the standard one.
Rather than use the minimal 15 single-qubit gate and 3 CNOT gate decomposition of $G$~\cite{vatan2004optimal, shende2004minimal} (aka the KAK decomposition~\cite{tucci2005introduction}), we use a slightly larger 21 single-qubit gate and 4 CNOT gate decomposition. Though our choice has more parameters than necessary, it is defined so that $G(\vec{0}) = I$, which is not true for the KAK decomposition. This has the desirable property that $W(\vec{\alpha} = \vec{0}) = I$. Of course, all risks/costs comparing $U$ and $V$ are defined so that when $W = V^{\dagger} U = I$, they vanish (i.e. $R = 0$ and $C = 0$ for any sensible risk $R$ and cost $C$). By writing $\vec{\alpha} = r \vec{\theta}$ we emphasize the point that regardless of the choice of $\vec{\theta}$, setting $r = 0$ samples the point $(0, 0)$ where both risks vanish.

\begin{figure}[H]
    \centering
    \includegraphics[width=0.6\textwidth]{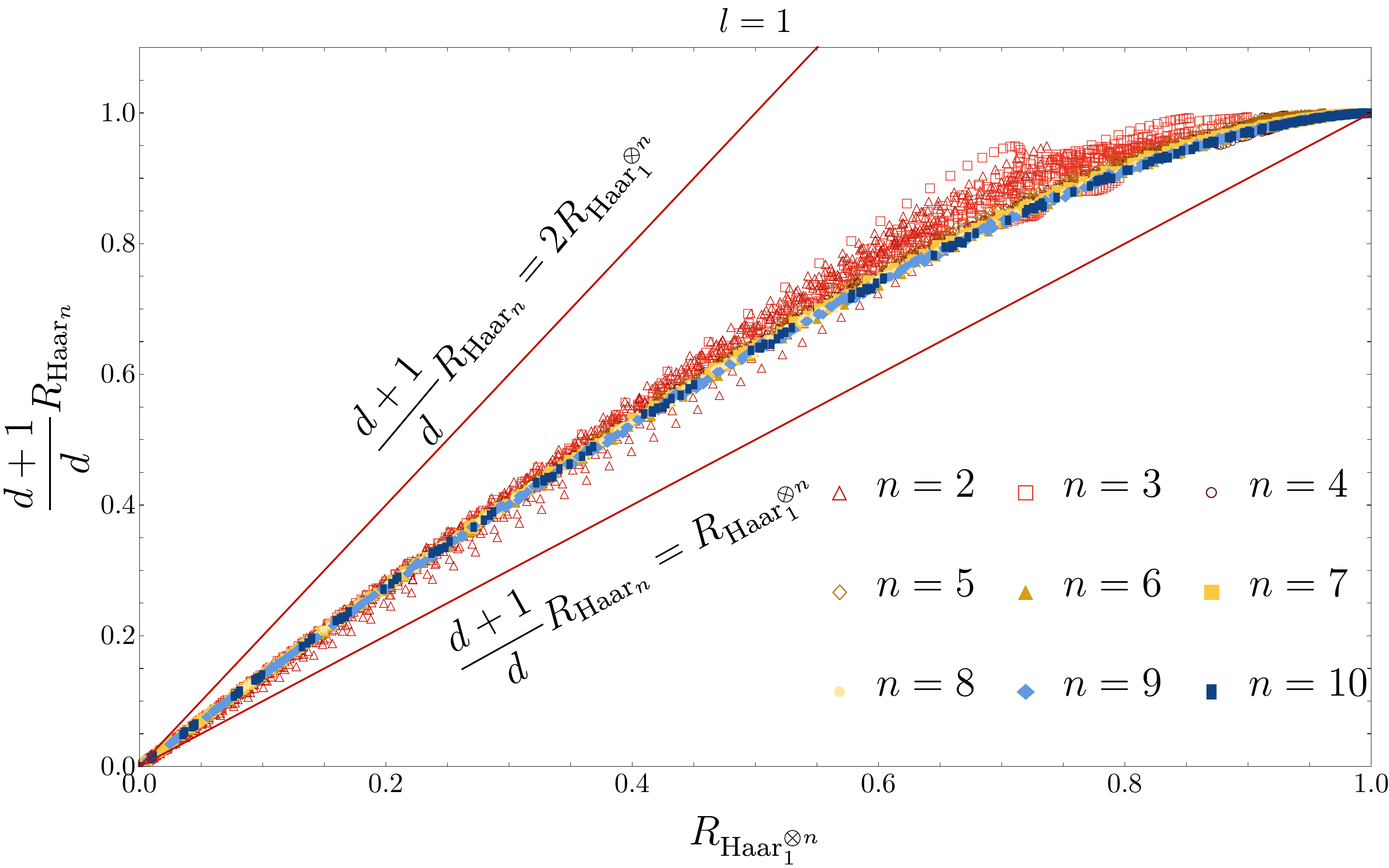}
    \includegraphics[width=0.6\textwidth]{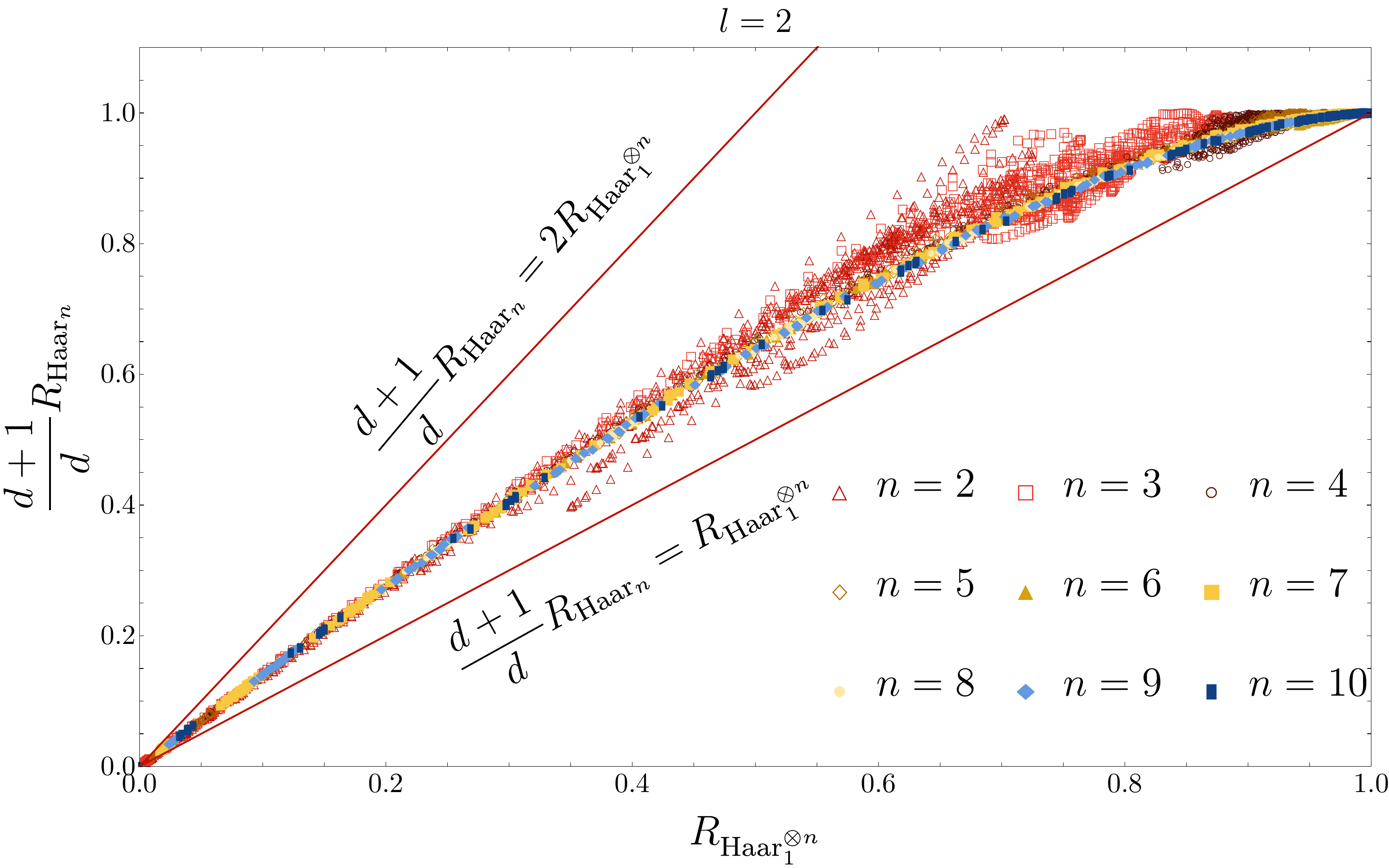}
    \includegraphics[width=0.6\textwidth]{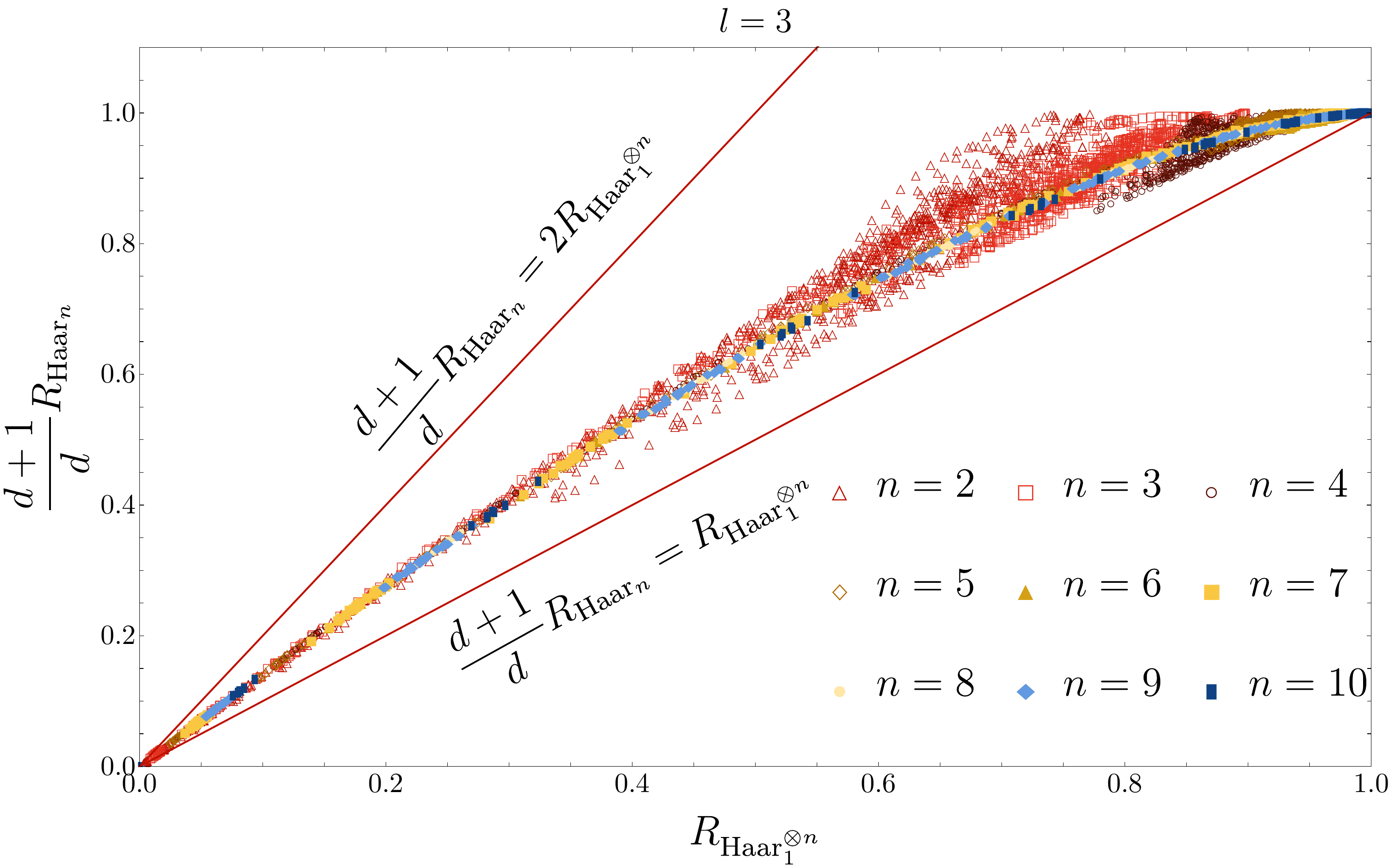}
    \caption{
    We provide numerical verification of Lemma~\ref{lem:locally-scrambled-cost-versus-hst-cost} for $\mathcal{Q} = \mathcal{S}_{{\rm Haar}_1^{\otimes n}}$, i.e. the ensemble of tensor products of Haar-random single-qubit states. In this case, the claim is that Ineq.~\eqref{eq:nisq-cost-comp} should be satisfied for any circuit parameters $\vec{\alpha}$. Graphically, a sampled point satisfies the claimed Ineq.~\eqref{eq:nisq-cost-comp} if it lies between the two reference lines, which as expected is true for all generated points. Hence, learning the action of a unitary on random product states is sufficient to generalize to the entire Hilbert space in the very practical setting of low depth variational quantum circuits as shown here. It is also interesting to note that the upper-bound is empirically not very tight in these generic variational quantum circuits. In fact, the scaling of the locally trained cost to the global cost is empirically linear after sufficient training.}
    \label{fig:cost-comp}
\end{figure}

By defining the ansatz in this way, we can easily sample a large range of cost values. In particular, for fixed $l$ and $ n$, we     randomly sample $(\vec{\theta_{ik}})_p \sim \mathcal{N}(0, 2 \pi) \ \forall i, k, p$, generating a random initial parameter vector $\vec{\theta}^{(0)}$ for the entire ansatz from which we then compute the starting risks, $(R_{\mathcal{S}_{{\rm Haar}_1^{\otimes n}}}(U, V(\vec{\theta}^{(0)})), \frac{d + 1}{d}R_{\mathcal{S}_{\text{Haar}_n}}(U, V(\vec{\theta}^{(0)})))$. Note that this is the same initialization procedure for many variational quantum algorithms (VQAs)~\cite{cirstoiu2020variational, gibbs2021long, cerezo2021cost}, and for a deep enough circuit with a random choice of $\vec{\theta}^{(0)}$, any risk or cost will be approximately maximized (i.e. $R \approx 1$ and $C \approx 1$). If we were truly running a VQA to learn $U(\vec{\theta})$ given a known unitary $V$, we would need to iteratively estimate $R_{\mathcal{S}_{{\rm Haar}_1^{\otimes n}}}$ on a quantum computer and update our best guess for $\vec{\theta}^*$ classically. Instead, we exploit the form of our toy ansatz directly: we simply re-scale each angle by $r$, i.e. $\vec{\theta}^{(0)} \rightarrow r \vec{\theta}^{(0)}$ for different values of $r$. By construction, the risk will be minimized when $r = 0$. Thus, by sampling values of $r \in [0, 1]$, we can explore the empirical cost relationship $(R_{\mathcal{S}_{{\rm Haar}_1^{\otimes n}}}(U, V(r \cdot \vec{\theta}^{(0)})), \frac{d + 1}{d}R_{\mathcal{S}_{\text{Haar}_n}}(U, V(r \cdot \vec{\theta}^{(0)})))$ between the two extremes. In Fig.~\ref{fig:cost-comp},  we show this empirical relationship for $n = 2, \ldots, 10$ qubits with ansatz depths $l =1, 2, 3$ for 20 random initialization vectors $\vec{\theta}^{(i)}$ and $100$ values $r$ for each random sample. For all points sampled, Ineq.~\eqref{eq:nisq-cost-comp} is satisfied as expected. In fact, the upper-bound is often pretty loose, and it appears that a tighter relationship might be true even for large values of $R_{\mathcal{S}_{{\rm Haar}_1^{\otimes n}}}(U, V(\vec{\alpha}))$.

\subsection{Out-of-Distribution Generalization for Learning Fast Scramblers}\label{ap:additional-numerics}

\textbf{Task and setup:} Here, we consider the task of learning a so-called fast scrambler~\cite{belyansky2020minimal} into an ansatz $V(\vec{\alpha})$ of a similar form.
An $n$-qubit fast scrambler unitary $U$ composed of $t$ time steps is defined as follows:
\begin{equation} \label{eq:Up1}
	U = \prod_{j=1}^{t} U^{\mathrm{I}}_j U^{\mathrm{II}}_j \ , 
\end{equation}
where $U^{\mathrm{I}}_j$ is a product of independent Haar-random single-qubit unitaries, $U^{\mathrm{I}}_j = \prod_{k = 1}^n u_{j,k}$, and $U^{\mathrm{II}}_j$ is given by
\begin{equation} \label{eq:Up2}
	U^{\mathrm{II}}_j = e^{-i \frac{g}{2\sqrt{n}} \sum_{k < \ell} Z_k Z_\ell} \ ,
\end{equation}
where $g$ is a real parameter.

The ansatz $V(\vec{\alpha})$ for learning the scrambler $U$ has the same structure as $U$ with fixed single-qubit gates replaced by parametrized ones. That is, the ansatz $V(\vec{\alpha})$ takes the form
\begin{equation}
	V(\vec{\alpha}) = \prod_{j=1}^{t} V^\mathrm{I}_j(\vec{\alpha}_j) U^{\mathrm{II}}_j \ ,
\end{equation}
where $V^\mathrm{I}(\vec{\alpha}_j) = \prod_{k=1}^n v_{j,k}(\vec{\alpha}_{j,k})$, with parametrized one qubit gates $v_{j,k}(\vec{\alpha}_{j,k})$. 
Here, we view $t$ as the number of time steps. This is a parameter that controls the difficulty of the optimization problem, since the entanglement introduced by $U$ quickly grows with $t$.
The parameter $g$ in Eq.~\eqref{eq:Up2} controls how quickly the learning difficulty grows with $t$. We work with $g=1$, but consider several values of $t$.

The learning is performed as described in the main text. That is, we first build a training set and then optimize a corresponding cost function. First, we generate a training set of size $N$ of the form $\mathcal{D}_{\mathcal{Q}}(N)=\{ \ket{\psi_j}, U\ket{\psi_j} \}_{j=1}^N$, where input states $\ket{\psi_j}$ are random product states. 
Second, we optimize the parameters $\vec{\alpha}$ according to the cost function $C_{\mathcal{D}_{\mathcal{Q}}(N)}(\vec{\alpha})$ introduced in Eq.~\eqref{eq:training-cost-general}.
Optimized parameters $\vec{\alpha}_\mathrm{opt}$ are found by (approximately) solving the optimization problem:
\begin{equation} \label{eq:opt}
    \vec{\alpha}_\mathrm{opt} = \operatorname{argmin}_{\vec{\alpha}} C_{\mathcal{D}_{\mathcal{Q}}(N)}(\vec{\alpha}) \ .
\end{equation}
We measure the learning quality with the (out-of-distribution) risk $R_{\mathcal{S}_{\rm Haar_n}}(\vec{\alpha})$, see Eq.~\eqref{eq:test-cost-general} and also Eq.~\eqref{eq:hst-cost-versus-expected-haar-cost}.
We are interested in generalization error $R_{\mathcal{S}_{\rm Haar_n}}(\Vec{\alpha}_\mathrm{opt}) - C_{\mathcal{D}_{\mathcal{Q}}(N)}(\vec{\alpha}_\mathrm{opt})$ as a function of various parameters in the learning problem.\\

\textbf{Results:} We learn an 8-qubit fast-scrambler unitary $U$ with $t = 3,\ldots,10$ and training data set sizes $N = 1,\ldots,15$. 
The learning is performed by repeating the optimization in Eq.~\eqref{eq:opt} 1000 times. 
Each optimization learns a different randomly generated $U$, works with different randomly drawn training set, and starts with different random initial parameters $\vec{\alpha}_0$. 
The optimization is performed with a variant of the gradient descent method in which the single-qubit unitaries $v_{j,k}(\vec{\beta})$ are spanned by three rotation angles, $v_{j,k}(\vec{\beta}) = e^{-iZ\beta_1} e^{-iX\beta_2} e^{-iZ\beta_3}$.

\begin{figure}[t]
\includegraphics[width=\textwidth]{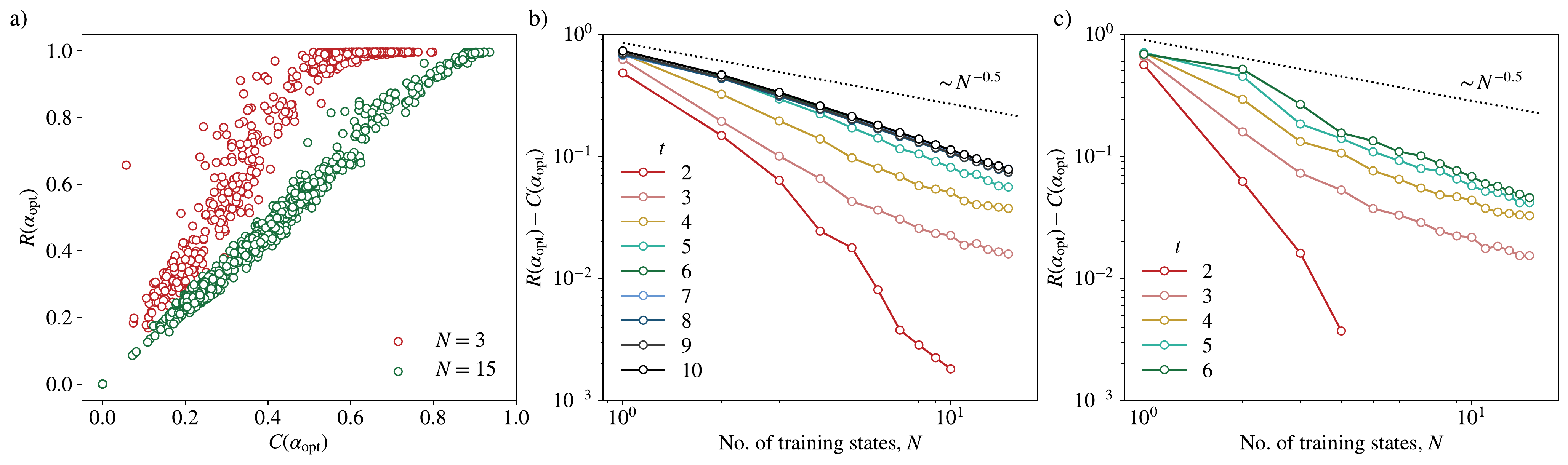}
\caption{\textbf{Learning an 8-qubit fast scrambler.} The individual panels extract scaling properties from the optimization data and show that our theoretical bounds apply, but are outperformed in practical learning tasks even for highly non-trivial, entangling unitaries. 
a) Testing risk $R_{\mathcal{S}_{\rm Haar_n}}(\Vec{\alpha}_\mathrm{opt})$ as a function of the training cost $C_{\mathcal{D}_{\mathcal{Q}}(N)}(\vec{\alpha}_\mathrm{opt})$ calculated for 1000 independently obtained values of $\vec{\alpha}_\textrm{opt}$.
b) Average generalization error $R_{\mathcal{S}_{\rm Haar_n}}(\Vec{\alpha}_\mathrm{opt}) - C_{\mathcal{D}_{\mathcal{Q}}(N)}(\vec{\alpha}_\mathrm{opt})$ over all 1000 optimizations as a function of $N$ for several values of $t$.
c) Average generalization error $R_{\mathcal{S}_{\rm Haar_n}}(\Vec{\alpha}_\mathrm{opt}) - C_{\mathcal{D}_{\mathcal{Q}}(N)}(\vec{\alpha}_\mathrm{opt})$ over successful optimization runs as a function of $N$ for several values of $t$.
See text for an in-depth analysis.}
\label{fig:sc}
\end{figure}

Figure~\ref{fig:sc} summarizes our results. Panel (a) shows the testing risk $R_{\mathcal{S}_{\rm Haar_n}}(\Vec{\alpha}_\mathrm{opt})$ as a function of the training cost $C_{\mathcal{D}_{\mathcal{Q}}(N)}(\vec{\alpha}_\mathrm{opt})$ calculated for 1000 independently obtained values of $\vec{\alpha}_\textrm{opt}$. 
The data was obtained for $t = 5$. 
Blue (red) dots represent optimization performed with training data size $N = 3$ ($N = 15$). 
We observe that small training data size of $N=3$ may lead to a situation in which optimization has already reached appreciable training cost values ($C_{\mathcal{D}_{\mathcal{Q}}(N)}(\vec{\alpha}_\mathrm{opt}) \approx 0.5$) but the testing risk is still at its maximal value ($R_{\mathcal{S}_{\rm Haar_n}}(\Vec{\alpha}_\mathrm{opt})\approx 1$). 
This generalization issue is resolved by adding more points to the training data set. Indeed, when training on $N=15$ data points, the obtained data does not display a concentration around $R_{\mathcal{S}_{\rm Haar_n}}(\Vec{\alpha}_\mathrm{opt}) \approx 1$. We also observe that larger data sets result in an increased likelihood of almost perfect learning (that is, an optimization that achieves $C_{\mathcal{D}_{\mathcal{Q}}(N)}(\vec{\alpha}_\mathrm{opt}) \simeq R_{\mathcal{S}_{\rm Haar_n}}(\Vec{\alpha}_\mathrm{opt}) \simeq 0$). Only 7\% of the optimization runs with $N=3$ achieved almost perfect learning while $12.5\%$ of the optimization runs with $N=15$ reached that goal. 
That plot also shows that larger training set leads to better generalization: achieving a given cost value of $C_{\mathcal{D}_{\mathcal{Q}}(N)}(\vec{\alpha}_\mathrm{opt})$ with a bigger training set results in smaller risk $R_{\mathcal{S}_{\rm Haar_n}}(\Vec{\alpha}_\mathrm{opt})$. We observe this behavior for every optimization that we have performed.

Panel (b) corroborates those findings further. It shows the generalization error $R_{\mathcal{S}_{\rm Haar_n}}(\Vec{\alpha}_\mathrm{opt}) - C_{\mathcal{D}_{\mathcal{Q}}(N)}(\vec{\alpha}_\mathrm{opt})$, averaged over all 1000 optimizations, as a function of $N$, the training data size, for several values of $t$. 
We see that the average generalization error obtained for the values $t = 6,\ldots,10$ is almost identical. The reason for this behavior is likely that for these values of $t$ and for the training data sizes $N$ used in our experiment, only very few optimization runs managed to lower the cost function enough to achieve a risk $R_{\mathcal{S}_{\rm Haar_n}}(\Vec{\alpha}_\mathrm{opt})$ smaller than its maximal value. 
This optimization issue might be dealt with by more refined minimization techniques. While the results in this setup (large $t$ and insufficiently large $N$) are not useful from a learning point of view, our theoretical generalization bounds still hold. As the data suggests, the scaling is better than the worst case scenario covered by the theoretical analysis.

Panel (c) avoids interpretational complications caused by optimization issues and averages only those minimization runs that achieved $C_{\mathcal{D}_{\mathcal{Q}}(N)}(\vec{\alpha}_\mathrm{opt}) < 0.5$. We see a scaling behavior similar to what we observed when taking the entire data into account. Panels (b) and (c) show that the generalization error decreases faster than theoretical upper bound, which scales as $\sim N^{-1/2}$ with the training data size $N$ and is shown by the black solid line. As the learning difficulty (measured by $t$) increases, the rate at which generalization error decreases with $N$ seems to slowly approach the theoretical bound.